\documentclass[11pt,a4paper]{article}

\usepackage[utf8]{inputenc}
\usepackage[T1]{fontenc}
\usepackage[english]{babel}
\usepackage{pgfgantt}
\usepackage{graphicx}
\usepackage{xcolor}
\usepackage{indentfirst}
\usepackage{url}
\usepackage{array}
\usepackage{multirow}
\usepackage{caption}
\usepackage{setspace}
\usepackage{breakcites}
\usepackage{float}
\usepackage{times}
\usepackage[hidelinks]{hyperref}
\usepackage{yfonts}
\usepackage{amsthm}
\usepackage{amssymb}
\usepackage{amsmath}
\usepackage{enumitem}
\usepackage[capitalise]{cleveref}
\usepackage{mathtools}
\usepackage{mathrsfs} %needed for the powerset "P" - \mathscr{P}
\usepackage[colorinlistoftodos]{todonotes}
\usepackage{bussproofs}
\usepackage{multicol}
\usepackage{soul,color}
\usepackage{tikz}
\tikzstyle{vertex}=[auto=left,circle,fill=black!25,minimum size=20pt,inner sep=0pt]
\theoremstyle{definition}
\newtheorem{defn}{Definition}[section]

\newtheorem{thm}[defn]{Theorem}
\newtheorem{cor}[defn]{Corollary}
\newtheorem{prop}[defn]{Proposition}
\newtheorem{lemma}[defn]{Lemma}

\theoremstyle{remark}
\newtheorem{rmk}{Remark}
\DeclareSymbolFont{symbolsC}{U}{txsyc}{m}{n}
\DeclareMathSymbol{\strictif}{\mathrel}{symbolsC}{74}

\makeatletter
\providecommand{\bigsqcap}{%
	\mathop{%
		\mathpalette\@updown\bigsqcup
	}%
}
\newcommand*{\@updown}[2]{%
	\rotatebox[origin=c]{180}{$\m@th#1#2$}%
}
\makeatother

\newcommand{\prim}{^\prime}
\newcommand{\cneg}{{\sim}}

\newcommand{\obl}{\textsf{O}}

\newcommand{\per}{\textsf{P}}
\newcommand{\defeq}{=_{\text{def}}}

\newcommand{\aland}{\tilde{\land}}
\newcommand{\alor}{\tilde{\lor}}
\newcommand{\ato}{\tilde{\to}}
\newcommand{\aneg}{\tilde{\neg}}

\newcommand{\aobl}{\tilde{\obl}}

\newcommand{\acirc}{\tilde{\circ}}

\newcommand{\nuciw}{v^{DmbCciw}}
\newcommand{\nuci}{v^{DmbCci}}
\newcommand{\nubc}{v^{DbC}}
\newcommand{\nuCi}{v^{DCi}}
\newcommand{\nucL}{v^{DmbCcl}}
\newcommand{\nucila}{v^{DCila}}
\newcommand{\nucone}{v^{C_1}}

\newcommand{\sdl}{\ensuremath{SDL}}
\newcommand{\mbC}{\ensuremath{mbC}}
\newcommand{\DbC}{\ensuremath{DbC}}
\newcommand{\DCi}{\ensuremath{DCi}}
\newcommand{\DCila}{\ensuremath{DCila}}
\newcommand{\Cila}{\ensuremath{Cila}}
\newcommand{\DCone}{\ensuremath{DC_1}}
\newcommand{\mbCcl}{\ensuremath{mbCcl}}
\newcommand{\mbCciw}{\ensuremath{mbCciw}}
\newcommand{\DmbC}{\ensuremath{DmbC}}
\newcommand{\DmbCciw}{\ensuremath{DmbCciw}}
\newcommand{\DmbCci}{\ensuremath{DmbCci}}
\newcommand{\DmbCcl}{\ensuremath{DmbCcl}}

\newcommand{\lfis}{\textit{LFI}s}
\newcommand{\ldis}{\textit{LDI}s}

\newcommand{\cplp}{\ensuremath{CPL^+}}
\newcommand{\cpl}{\ensuremath{CPL}}
\newcommand{\A}{\ensuremath{\mathcal{A}}}
\newcommand{\B}{\ensuremath{\mathcal{B}}}
\newcommand{\D}{\ensuremath{\mathcal{D}}}
\newcommand{\U}{\ensuremath{\mathcal{U}}}
\newcommand{\M}{\ensuremath{\mathcal{M}}}

\newcommand{\axK}{\textsf{(K)}}
\newcommand{\axD}{\textsf{(D)}}

\title{Swap Kripke models for deontic \lfis}
\author{Mahan Vaz\footnote{Instituto de Filosofia e Ciências Humanas (IFCH), Universidade Estadual de Campinas (UNICAMP), Brazil, and Institut für Philosophie I, Logik und Erkenntnistheorie, Ruhr-Universität Bochum, Germany. E-mail: \texttt{Mahan.VazSilva@edu.ruhr-uni-bochum.edu}} \ and  Marcelo Esteban Coniglio\footnote{Instituto de Filosofia e Ci\^{e}ncias Humanas (IFCH), and
Centro de L\'ogica, Epistemologia e Hist\'oria da Ci\^encia (CLE),  Universidade Estadual de Campinas (UNICAMP), Brazil. E-mail: \texttt{coniglio@unicamp.br}}}

\begin{document}
	\maketitle

\begin{abstract}
    We present a construction of nondeterministic semantics for some deontic logics based on the class of paraconsistent logics known as {\em Logics of Formal Inconsistency} (\lfis), for the first time combining swap structures and Kripke models through the novel notion of swap Kripe models. We start by making use of Nmatrices to characterize systems based on \lfis\ that do not satisfy axiom \textbf{(cl)}, while turning to RNmatrices when the latter is considered in the  underlying \lfis. This paper also presents, for the first time, a full axiomatization and a semantics for the $C^{D}_n$ hierarchy, by use of the aforementioned mixed semantics with RNmatrices. This includes the historical system $C^{D}_1$ of da Costa-Carnielli (1986), the first deontic paraconsistent system proposed in the literature.
    
    %Besides we expand on previous works by Coniglio and Peron (2009) and Coniglio and Toledo (2022), by applying the mixed semantics to systems satisfying not only the modal deontic axioms, but also the axiom \textbf{(cl)}. Further, we present for the first time a characterization and a semantics for the whole $\textbf{C}^D_n$ hierarchy, development which was only hinted at by da Costa and Carnielli (1986), but never fully explicitly given.   
\end{abstract}

    {\bf Keywords:} Deontic logic; paraconsistent logic; nondeterministic semantics;

\section{Introduction}

One proposal to deal with deontic paradoxes which derive a conflict of obligations was by means of changing the basic logic to a paraconsistent one.
This was pioneered by Carnielli and da Costa~\cite{CostaCarnielli1986}, when the authors proposed the system $C_1$ as the base logic for a deontic system. The idea behind this move was simple: given a formula $\alpha$ in $C_1$, the occurrence of $\alpha \land \neg \alpha$ does not trivialize the system. Hence, extending it to a modal system with the axioms of Standard Deontic Logic would also prevent such a system to be trivialized. Moreover, since $C_1$ is the first logic in a hierarchy of logics $C_n$, for $n \in \mathbb{N}$, then the move could be similarly applied to any logic in the hierarchy. The latter, however, has hitherto never been done. In addition, a formal semantics treatment for such systems  has never been  provided in the literature. One of the aims of this  paper is to fill these gaps.

With the passing years, the notion of deontic paraconsistency has evolved. Not only other authors worked on the subject \cite{McGinnis2007, BeirlanStrasser2011}, but in particular, Coniglio and Peron have developed and studied the systems named \ldis, which expand the notion of paraconsistent deontic logics developed by da Costa and Carnielli, now applying the axioms of Standard Deontic Logic to many {\lfis}.\footnote{Also to some other logics, such as Batens' CLuN, previously named DPI \cite{Batens1980a, Batens1980b}.} This also brought the notion of deontic paraconsistency to light; a logic is \textit{deontically paraconsistent} if it is not deontically explosive, i.e., for some $\alpha, \beta$ in the set of formulas, we have the following:  $$\obl \alpha, \obl \neg \alpha \nvdash \obl \beta$$

Moreover, a logic is a Logic of Deontic Inconsistency, \textit{LDI}, if it is not deontically explosive and, besides, there  a unary connective (primitive or defined) $\bar{\boxminus}$ such that the following holds:\footnote{As in the case of \lfis, in the general case $\bar{\boxminus}(p)$ can be considered as being a set of modal formulas depending on a single  propositional letter $p$. \ldis\ where introduced by Coniglio in~\cite{Coniglio2009b}. Additional developments and applications of  \ldis\ can be found in \cite{Coniglio2009}.}
\begin{itemize}
        \item For some sentences $\alpha, \alpha\prim,\beta,\beta\prim$,
        \begin{itemize}
            \item $\bar{\boxminus}(\alpha), \obl\alpha \nVdash \obl\beta $
            \item $\bar{\boxminus}(\alpha\prim), \obl\neg\alpha\prim \nVdash \obl\beta\prim $
        \end{itemize}  
        \item For any $\alpha, \beta$
        \begin{itemize}
            \item $\bar{\boxminus}(\alpha), \obl\alpha, \obl \neg \alpha \Vdash \obl\beta $
        \end{itemize}
\end{itemize}

Any normal modal logic based on an \textit{LFI} can be seen as an \textit{LDI} simply by taking $\bar{\boxminus}(\alpha):=\obl{\circ}\alpha$. The original semantics for \ldis\ based on \lfis\ provided in~\cite{Coniglio2009b,Coniglio2009} was given in terms of Kripke structures together with bivaluation semantics.\footnote{It is worth noting that Bueno-Soler has introduced a wide class of  paraconsistent modal systems based on \lfis, also with a semantics given by  Kripke structures equipped with bivaluation semantics, and alternatively with a possible-translations semantics. See, for instance, \cite{Bueno-Soler10}.} 

The reappearance of nondeterministic matrix (Nmatrix) semantics also sparked the interest in modal logics. The recent years have seen a rise in developments in the area, with prominent works (in chronological order) by Coniglio, Fari\~nas del Cerro and Peron~\cite{Coniglio2015},  Skurt and Omori~\cite{OmoriSkurt16}, Coniglio and Golzio~\cite{Coniglio19}, Grätz~\cite{Graetz2021}, Pawlowski and Skurt~\cite{Pawlowski2024}, Pawlowski, Coniglio and Skurt~\cite{ConPawSku2024} and Leme et al.~\cite{leme2025}, among others. In general, these works show that Nmatrices and RNmatrices (i.e., restricted Nmatrices) allow for the characterization of many non-normal modal logics.  These results motivated the aim to tackle deontic \lfis\ nondeterministically and the pursue to cover the whole $C_n$ hierarchy, as envisioned by Carnielli and da Costa in 1986. 

It is important to note that, although the works of Coniglio and Peron, as well as Bueno-Soler, cover a portion of the \lfis' hierarchy, some of the \lfis\  were not studied at the time. Pertaining to the latter were the systems satisfying axiom {\bf (cl)}, which are not characterizable by finite Nmatrices, by the Dugundji-like theorem proved by Arnon Avron \cite[Theorem~11]{Avron2007a}. In order to characterize them one must resort to RNmatrices \cite{ConiglioToledo2022}, with their respective modal versions being rather complex and needing an extra layer of caution when being dealt with. Despite all the caution, the results by Coniglio and Toledo present a new possibility for a nondeterministic semantic characterization of these logics.

Having these details in mind, this paper starts by presenting a semantics for some deontic \lfis, starting with \DmbC. The novelty at this point is that, different to the previous approaches to modal \lfis\ found in the literature,  we present a novel semantics given by a combination  between swap structures and Kripke models, by allowing sets of worlds and relations as a frame, where each of the worlds is nondeterministic. We then move to extensions of this logic eventually reaching \DmbCcl, which, by previous known results, is not characterizable by finite Nmatrices. We then show, following the results presented in \cite{ConiglioToledo2022}, that the combination between RNmatrices and Kripke models allows for a characterization of this logic, as well as its extensions, namely \DCila, its conservative reduct $C^{D}_1$ and the hierarchy extension $C^{D}_n$, for all $C_n$. Regarding the latter, we present for the first time an explicit characterization of these logics, describing its axioms and a semantics (once again in terms of swap Kripke models), given that the deontic systems for the whole hierarchy were never explicitly described. We end this paper with a brief discussion of our results, applying the $C^{D}_n$ hierarchy to moral dilemmas.

\section{The paraconsistent deontic system \DmbC}
    \label{sec:DmbC}
    We give a modal account of the fundamental \textit{LFI}, {\mbC} together with a modalization proposed in \cite{Coniglio2009b, Peron2008, Coniglio2009}, which is a \textit{deontic} version of $\mbC$, which we call \DmbC.

\begin{defn} \label{def:DmbC}
	Let $\Sigma = \{\to , \neg, \lor, \land, \obl,\circ\}$ be a signature for {\em LFI}s. The logic {\DmbC} defined over $\Sigma$ is the system characterized by all CPL$^+$ axioms, that is, the axioms corresponding to the positive fragment of classical propositional logic, plus the following axioms for $\neg$ and $\circ$:
	
	\begin{description}
		\item[(EM)] $\alpha \lor \neg \alpha$
		\item[(bc)] $\circ\alpha \to (\alpha \to (\neg \alpha \to \beta))$
	\end{description}
	
\noindent	together with the following modal axioms: 
	
	\begin{description}
		\item[($\obl-K$)] $\obl(\alpha \to \beta) \to (\obl \alpha \to \obl \beta)$ 
		\item [($\obl-E$)] $ \obl\bot_\alpha \to \bot_\alpha$, where $\bot_\alpha := (\alpha \land \neg \alpha) \land \circ \alpha $
	\end{description}
	
\noindent		such that the only inference rules are Modus Ponens and $\obl$-necessitation.
\end{defn}

Observe that $\obl$-necessitation is a {\em global} inference rule (i.e., it only can be applied to premises which are theorems). From this, the notion of {\em derivation from premises} needs to be adjusted in \DmbC, as it is usually done in normal modal systems.

\begin{defn} \label{def:derivations}
Let $\Gamma \cup \{\varphi\} \subseteq For(\Sigma)$.\\[1mm]
(1) A {\em derivation} of $\varphi$ in \DmbC\ is a finite sequence of formulas $\varphi_1 \ldots\varphi_n$ such that $\varphi_n=\varphi$ and, for every  $1 \leq i \leq n$, either $\varphi_i$ is an instance of an axiom, or $\varphi_i$ follows from $\varphi_j$ and $\varphi_k=\varphi_j \to \varphi_i$ (for $j,k<i$) by Modus Ponens, or $\varphi_i=\obl\varphi_j$ follows from $\varphi_j$ (for $j<i$) by  $\obl$-necessitation. In this case, we say that $\varphi$ is {\em derivable}  in \DmbC, or it is a {\em theorem} of \DmbC, which will be denoted by $\vdash_{\DmbC} \varphi$.\\[1mm]
(2) We say that {\em $\varphi$ is derivable from $\Gamma$ in \DmbC}, denoted by  $\Gamma\vdash_{\DmbC} \varphi$, if either  $\vdash_{\DmbC} \varphi$, or there exist formulas $\gamma_1,\ldots,\gamma_k \in \Gamma$ (for a finite $k \geq 1$) such that $\vdash_{\DmbC} (\gamma_1  \land\ldots \land \gamma_k) \to \varphi$ .
\end{defn} 

Observe that $\emptyset\vdash_{\DmbC} \varphi$ iff $\vdash_{\DmbC} \varphi$.
We know that the system \DmbC\ is metatheoretically well-behaved, as the results by Coniglio \cite{Coniglio2009b} show, where both the deduction metatheorem and proof-by-cases is shown to hold for \DmbC. 

It is important to note a few things about the conception of this logic. We define $\bot_\alpha$ being equivalent to  $(\alpha \land \neg \alpha) \land \circ \alpha$. Since \mbC\ is a minimal \textit{LFI} \cite{Marcos2007}, it contains the consistency operator $\circ$ and it serves the purpose of saying when a certain formula $\alpha$ is metatheoretically well-behaved in the system from a \textit{logical} perspective. 
By taking $\cneg \alpha :=(\alpha \to \bot_\alpha)$, we recover classical negation and can again define permission in terms of $\cneg$. This allows us to add $(\obl-E)$ instead of the usual deontic axiom
\begin{description}
    \item[$(\obl-D)$] $\obl\alpha \to \per\alpha $
\end{description}
 where $ \per\alpha := \cneg \obl\cneg \alpha$, for characterizing \DmbC. Observe that, as usual, $ \per\alpha$ means that `$\alpha$ is permitted'.
 
So let us take $\bot$ to be a bottom formula in \cpl. By our definition of $\per\alpha$, the following result ensues: $$\obl\alpha \to \cneg \obl\cneg \alpha \equiv \obl\alpha \to (\obl\cneg \alpha \to \bot) \equiv (\obl\alpha \land \obl\cneg \alpha) \to \bot$$ and given that $\obl\alpha \land \obl\cneg \alpha \equiv \obl(\alpha \land \cneg \alpha) $, then we get $\obl(\alpha \land \cneg \alpha) \to \bot$ or, equivalently, $\cneg\obl(\alpha \land \cneg \alpha)$ (another standard way to represent the deontic axiom). In turn, if we define $f_\alpha := (\alpha \land \cneg \alpha)$, then the last result is equivalent to $\obl f_\alpha \to f_\alpha$. The equivalence used to obtain the last result is target for many criticisms in deontic logics, however, it will not be within the scope of this paper to address such criticisms.

    \subsection{Swap Kripke models for \DmbC}
    
    \textit{Swap} structures are multialgebras of a particular kind,  defined over ordinary algebras, and whose domains are the truth values of a certain logic, but presented as finite sequences of values of the underlying algebra. These sequences, called {\em snapshots},  represent (semantical) states of a given formula, described by the components of the sequence. For \DmbC, the snapshots consists of pairs  over the two-element Boolean algebra with domain $2=\{0,1\}$ representing the semantical state of a formula and of its paraconsistent negation $\neg$. The {\em consistency} (or {\em classicality}) operator $\circ$ is defined in terms of its relation with contradiction w.r.t. the paraconsistent negation. 
  
	\begin{defn}
    \label{def:Nmatrix-DmbC}
	Let $\mathcal{A}_3:= \langle A, \tilde{\land}, \tilde{\lor}, \tilde{\neg}, \tilde{\to}, \tilde{O}, \tilde{\circ}\rangle$ be a multialgebra with domain $A = \{T,t,F\}$.  Let $\mathcal{D} = \{T,t\} $ denote the designated truth values and define $\mathcal{M}_3:= \langle \mathcal{A}_3, \mathcal{D} \rangle$ to be an $N$matrix over signature $\Sigma$.
	\end{defn} 
	
	\begin{rmk}
    \label{rmk:values-DmbC}
	As mentioned above, the domain of the multialgebras in which we are interested is formed by pairs over 2 intending to represent the (simultaneous) values in $2=\{ 0 ,1\}$ assigned to the formulas $\varphi$ and  $\neg\varphi$. That domain is the set of truth values for such logic. Thus, we let $A \subseteq 2^2$ and define $T := (1,0)$ ($\varphi$ is true, $\neg\varphi$ is false); $t:= (1,1)$ ($\varphi$ is true, $\neg\varphi$ is true); and $F:= (0,1)$ ($\varphi$ is false, $\neg\varphi$ is true). We remove from the domain the pair $(0,0)$  ($\varphi$ is false, $\neg\varphi$ is false) since the paraconsistent negation is assumed to satisfy the excluded-middle law (EM) (recall Definition~\ref{def:DmbC}). From now on, we mention whenever possible the snapshots instead of their labels. Notice that $\D = \{(1,0), (1,1)\}=\{z \in A \ : \ z_1=1 \}$. %\add{rephrase}
	\end{rmk}
	
	\begin{defn}  \label{def:swapOperations}
		The modal swap structure for {\DmbC} is $\mathcal{A}_3$ (cf. \cref{def:Nmatrix-DmbC}) such that its domain is $ \mathcal{B}_{\mathcal{A}_3}^{DmbC}=\{(c_1,c_2) \in A : c_1 \sqcup c_2 = 1\}$ and the multioperations $\tilde{\land}, \tilde{\lor}, \tilde{\to}, \tilde{\neg}, \tilde{\circ}$, as well as a special multioperator $\tilde{O}:\wp_+(A) \to \wp_+(A)$,\footnote{In this paper, $\wp_+(Y)$ will denote the set of non-empty subsets of a set $Y$.} are defined as follows, for every $a, b \in A $ and  $\emptyset \neq X \subseteq A$:
		
		\begin{enumerate}
			\item $a \tilde{\land} b := \{(c_1, c_2) \in A \ : \ c_1 = a_1 \sqcap b_1 \}$
            %, where $\# \in \{ \land, \lor, \to\}$
            \item $a \tilde{\lor} b := \{(c_1, c_2) \in A \ : \  c_1 = a_1 \sqcup b_1 \}$
            \item $a \tilde{\to} b := \{(c_1, c_2) \in A \ : \ c_1 = a_1 \supset b_1 \}$
			\item $\aneg a := \{ (c_1, c_2) \in A \ : \ c_1 = a_2 \}$
            \item $\acirc a := \{(c_1,c_2) \in A \ : \ c_1 \leq \cneg(a_1 \sqcap a_2)\}$
			\item $\aobl(X) := \{ (c_1, c_2) \in A\ : \ c_1 = \bigsqcap \{x_1 \ : \ x \in X\}\} $ .
		\end{enumerate}
        
	\end{defn}
  
    \begin{rmk}
        The symbols $\sqcap, \sqcup, \supset, \cneg$ refer to the Boolean operations of meet, join, implication and boolean complement in 2, respectively. The symbol $\bigsqcap$ is applied to a non-empty subset of $2$ and denotes the meet of all the elements of that set.
    \end{rmk}

	\begin{rmk}
		$\B_{\A_3}^{\DmbC} = A$ via the analytical representation of the truth values, shown in the previous remark. The non-deterministic truth-tables for the non-modal operators are displayed below, where $\U = \{F\} $ is the set of non-designated truth values.

\begin{center}

\begin{tabular}{|l|c|c|r|}
\hline
$\tilde{\wedge}$ & $T$ & $t$ & $F$ \\ \hline
$T$ & $\D$ & $\D$ & $\U$ \\ \hline
$t$ & $\D$ & $\D$ & $\U$ \\ \hline
$F$ & $\U$ & $\U$ & $\U$ \\ \hline
\end{tabular}
\hspace{2cm}
\begin{tabular}{|l|c|c|r|}
\hline
$\tilde{\vee}$ & $T$ & $t$ & $F$ \\ \hline
$T$ & $\D$ & $\D$ & $\D$ \\ \hline
$t$ & $\D$ & $\D$ & $\D$ \\ \hline
$F$ & $\D$ & $\D$ & $\U$ \\ \hline
\end{tabular}
\end{center}

\begin{center}
\begin{tabular}{|l|c|c|r|}
\hline
$\tilde{\rightarrow}$ & $T$ & $t$ & $F$ \\ \hline
$T$ & $\D$ & $\D$ & $\U$ \\ \hline
$t$ & $\D$ & $\D$ & $\U$ \\ \hline
$F$ & $\D$ & $\D$ & $\D$ \\ \hline
\end{tabular}
\hspace{1.5cm}
\begin{tabular}{|l|r|}
\hline
 & $\tilde{\neg}$ \\ \hline
$T$ & $\U$ \\ \hline
$t$ & $\D$\\ \hline
$F$ & $\D$\\ \hline
\end{tabular}
\hspace{1.5cm}
\begin{tabular}{|l|r|}
\hline
 & $\tilde{\circ}$ \\ \hline
$T$ & $A$ \\ \hline
$t$ & $\U$\\ \hline
$F$ & $A$\\ \hline
\end{tabular}

\end{center}

	\end{rmk}
    
	\begin{defn}
		\label{def:nu_w}
			Let $W$ be a non-empty set (of {\em possible worlds}), and $R \subseteq W^2$ be an {\em accessibility} relation on $W$. For each $w, w^\prime \in W$ a function $v_w : For(\Sigma) \to A$ is a \textit{swap valuation} for {\DmbC} if every condition below is satisfied, for $\alpha,\beta\in For(\Sigma)$:
			
			\begin{description}
			  %  \item[1.] $v_w (p) \in A$ is arbitrary, for a propositional variable $p$
				\item[1.] $v_w(\alpha \land \beta) \in v_w(\alpha) \aland v_w(\beta)$				
                \item[2.] $v_w(\alpha \lor \beta) \in v_w(\alpha) \alor v_w(\beta)$
                \item[3.] $v_w(\alpha \to \beta) \in v_w(\alpha) \ato v_w(\beta)$
				\item[4.] $v_w(\neg \alpha) \in \aneg v_w(\alpha)$								
				\item[5.] $v_w(\circ\alpha) \in \acirc v_w(\alpha)$
                \item[6.] $v_w(\obl\alpha) \in \aobl\big(\{v_{w\prim}(\alpha) \ : \ w R w^\prime\}\big)$ 
			\end{description}
	\end{defn}

    \begin{defn} \label{def:swap-Kripke-DmbC}
    Let $W$ be a non-empty set of worlds, $R \subseteq W^2$ be a serial accessibility relation\footnote{Recall that a relation $R \subseteq W^2$ is {\em serial} if, for every $w \in W$, there exists $w'\in W$ such that $wRw'$.} and $\{v_w\}_{w \in W}$ a family of swap valuations for \DmbC.  We say that the triple $\M^{K-S}_{\DmbC} = \langle W,R,\{v_w\}_{w \in W} \rangle$ is a {\em swap Kripke model} for the logic {\DmbC}.
    \end{defn}

   \begin{rmk}
   \label{rmk:notation} Let $i \in \{1,2\}$ and $w \in W$.  We define $\pi_i(v_w(\alpha))$ to be the projection of the pair $v_w(\alpha)$ on its $i$-th coordinate. For the sake of simplicity, we adopt the notation $\alpha_{(i,w)}$ to denote $\pi_i(v_w(\alpha))$. %Whenever $w$ is clear from context, we simply write $\alpha_i$.
    \end{rmk} 
    
    \begin{lemma}
    \label{lemma:proj-nu_w}
        Let $w \in W$ and $\alpha,\beta \in For(\Sigma)$. Moreover, let $v_w$ be as in \cref{def:nu_w}. Then 
            \begin{description}
                \item[1.] $(\alpha \land \beta)_{(1,w)} = \alpha_{(1,w)} \sqcap \beta_{(1,w)}$
                \item[2.] $(\alpha \lor \beta)_{(1,w)} = \alpha_{(1,w)} \sqcup \beta_{(1,w)}$
                \item[3.] $(\alpha \to \beta)_{(1,w)} = \alpha_{(1,w)} \supset \beta_{(1,w)}$
                \item[4.] $(\neg \alpha)_{(1,w)} = \alpha_{(2,w)}$
                \item[5.] $(\circ \alpha)_{(1,w)} \leq {\sim} (\alpha_{(1,w)} \sqcap \alpha_{(2,w)})$ 
                \item[6.] $(\obl \alpha)_{(1,w)} = \bigsqcap \{\alpha_{(1,w\prim)} \ : \ wRw\prim\}$ %\add{vielleicht $\leq$?}
            \end{description}
   \end{lemma}
    \begin{proof}
        Items 1. through 5. are immediate from the definitions above. For item 6., consider $\aobl\big(\{v_{w\prim}(\alpha) \ : \ w R w^\prime\}\big)$. By \cref{def:swapOperations}, $$\aobl\big(\{v_{w\prim}(\alpha) \ : \ w R w^\prime\}\big) = \{c \in A \ : \ c_1 = \bigsqcap \{x_{(1,w\prim)} \ : \ x \in \{v_{w\prim}(\alpha) \ : \ w R w^\prime\} \}\}$$
        thus our result follows by item 6. of \cref{def:nu_w}.
    \end{proof}

   \begin{rmk}
    \label{prop:designated-dmbc}
        By the very definitions, for any $w \in W$ and for any $\alpha \in For(\Sigma)$, $v_w(\alpha) \in \D$ if and only if $\alpha_{(1,w)} = 1$.
   \end{rmk}

    \begin{rmk} \label{rem:O-E}
        We can now give a clear picture of what our models will look like. In particular, we show how any model of $\DmbC$ satisfy axiom ($\obl$-E). It is easy to see  that $v_w(\obl\bot_\alpha)=v_w(\bot_\alpha)$. Indeed, by item~6 of Lemma~\ref{lemma:proj-nu_w},  $(\obl \bot_\alpha)_{(1,w)} = \bigsqcap \{(\bot_\alpha)_{(1,w\prim)} \ : \ wRw\prim\} =  \bigsqcap \{0 \ : \ wRw\prim\}=0$, given that $(\bot_\alpha)_{(1,w\prim)}=0$ for every $w\prim  \in W$ (by \cref{lemma:proj-nu_w} items 1, 4 and 5), and $\{w\prim \in W \ : \  w R w\prim\}\neq\emptyset$, since $R$ is serial. 
    \end{rmk}

    \begin{defn}
        \label{def:truth}
		Let $\mathcal{M}^{K-S}_{DmbC} = \langle W, R, \{v_w\}_{w \in W}  \rangle$ be as in \cref{def:swap-Kripke-DmbC}. For any formula $\alpha \in For(\Sigma)$, we say that $\alpha$ is $\mathcal{M}^{K-S}_{DmbC}$-true in a world $w$, denoted $\mathcal{M}^{K-S}_{DmbC}, w \vDash \alpha$, if $v_w(\alpha) \in \mathcal{D}$. 
	\end{defn}

	\begin{defn}
        \label{def:l-cons}
        Let $\Gamma \cup\{\alpha\} \subseteq For(\Sigma)$. We say that $\alpha$ is a logical consequence of $\Gamma$ in \DmbC, denoted $\Gamma \vDash_{\DmbC}  \alpha$, if for all $\mathcal{M}^{K-S}_{DmbC}$ and $w \in W$: $\mathcal{M}^{K-S}_{DmbC}, w \vDash \Gamma$ implies that $\mathcal{M}^{K-S}_{DmbC}, w \vDash \alpha$.
	\end{defn}

    \begin{thm}[Soundness of \DmbC\ w.r.t. swap Kripke models] \ \\
        For every $\Gamma \cup \{\varphi\} \subset For(\Sigma)$, if $\Gamma \vdash_{\DmbC} \varphi$, then $\Gamma \vDash_{\DmbC}  \varphi$.
    \end{thm}

    \begin{proof}
        We first show that the theorem holds for the axioms of {\DmbC}.  The result for the {\cplp} axioms follows from \cref{lemma:proj-nu_w}. 
        
        For (bc), we must show that $(\circ \alpha)_{(1,w)} \ato (\alpha \to (\neg \alpha \to \beta))_{(1,w)}=1$. If $(\circ \alpha)_{(1,w)}=0$, we are done.  Otherwise, ${\sim}(\alpha_{(1,w)} \land \alpha_{(2,w)})=1$, so either $\alpha_{(1,w)}=0$ or $\alpha_{(2,w)}=0$. The first guarantees that $(\alpha \to (\neg\alpha \to \beta))_{(1,w)} = 1$. The second makes it so that $(\neg \alpha \to \beta)_{(1,w)}=1$, hence $(\alpha \to (\neg\alpha \to \beta))_{(1,w)} = 1$.
        
        For (\obl-K), assume $(\obl(\alpha \to \beta))_{(1,w)}=1$ and that $(\obl(\alpha))_{(1,w)} = 1$. We then have that $(\alpha \to \beta)_{(1,w\prim)}=1$ and that $\alpha_{(1,w\prim)}=1$ for every $w\prim \in W$ such that $wRw\prim$. Hence, it follows that $\beta_{(1,w\prim)}=1$ for every $w\prim$ such that $wRw\prim$.

        The proof for (\obl-E) follows from Remark~\ref{rem:O-E}.
        %, suppose $(\obl(\bot_\alpha))_{(1,w)}=1$. Then $(\bot_\alpha)_{(1,w\prim)}=1$ for every $w\prim \in W$ such that $wRw\prim$. But this is impossible, by \cref{lemma:proj-nu_w} items 1, 4 and 5.

        For Modus Ponens, it is immediate to see that it satisfies the criteria, by definition of $\tilde{\to}$. For Necessitation, suppose $\alpha$ is a theorem. Then for every $w \in W$, $\alpha_{(1,w)}=1$. In particular, it is the case for every $w\prim \in W$ such that $wRw\prim$,  from which it follows that $(\obl\alpha)_{(1,w)}=1$.
    \end{proof}

    In order to prove completeness for {\DmbC}, we build canonical models based on swap structures. We use the method of $\psi$-saturation for construction of maximal consistent sets, together with the denecessitation for the accessibility relation. This will allow us to create the models for which completeness results will hold. 

     \begin{defn}
         Given a Tarskian and finitary\footnote{A logic {\bf L} is Tarskian if for $\Gamma \subseteq For(\Sigma)$ the following holds:\\
         {\bf (Reflexivity)} for every $\varphi \in \Gamma$, $\Gamma \vdash \varphi$;\\
         {\bf (Monotonicity)} if $\Gamma \vdash \varphi$ and $\Gamma \subseteq \Delta$, then $\Delta \vdash \varphi$;\\
         {\bf (Cut)} if $\Gamma \vdash \varphi$ for every $\varphi \in \Delta$ and $\Delta \vdash \psi$, then $\Gamma \vdash \psi$. \\
         {\bf L} is finitary if it satisfies\\
         {\bf (Finiteness)} $\Gamma \vdash \varphi$ implies $\Gamma_0 \vdash \varphi$, for some finite $\Gamma_0 \subseteq \Gamma$.} logic {\bf L}, a set of formulas $\Delta$ is  $\psi$-saturated in {\bf L} if $\Delta  \nvdash \psi$ and, if $\varphi \notin \Delta$, then $\Delta \cup \{\varphi\} \vdash \psi$.
     \end{defn}

      \begin{rmk}
      \label{prop:phi-sat}
         It is well-known that any  $\psi$-saturated set is a closed theory. Moreover, if $\Gamma \nvdash \psi$ in {\bf L} then there exists a set $\Delta$ which is $\psi$-saturated in {\bf L} and contains $\Gamma$. In particular, this property holds for \DmbC\ and all the other logics to be considered in this paper.
       \end{rmk}

     \begin{defn}
		\label{def:W-can}
		Consider the set $$W_{can} = \{\Delta \subseteq For(\Sigma) \ : \ \Delta \text{ is a } \psi\text{-saturated in \DmbC, for some $\psi \in For(\Sigma)$}\}.$$
	\end{defn}
 
	\begin{defn}
		Let $Den(\Delta):= \{\varphi \in For(\Sigma) \ : \ O\varphi \in \Delta\}$
	\end{defn}
	
	\begin{defn}
		\label{def:R-can}
		Let $R_{can}\subseteq W \times W$ be given by:  $$\Delta R_{can} \Theta \text{ iff } Den(\Delta) \subseteq \Theta$$ for $\Delta, \Theta \in W$
	\end{defn}
	
	\begin{defn}
    \label{def:can-bival-truth}
		For each $\Delta \in W_{can}$, let $v_\Delta: For(\Sigma) \to \A_3$ defined as follows:
        \begin{equation}
            v_\Delta(\alpha) = 
            \begin{dcases}
                T,& \text{if } \alpha \in \Delta, \neg\alpha \notin \Delta\\
                t,& \text{if } \alpha, \neg\alpha \in \Delta\\
                F,& \text{if } \neg\alpha \in \Delta, \alpha \notin \Delta 
            \end{dcases}
        \end{equation}  
	\end{defn}

    \begin{lemma}
        \label{lemma:truth-sat}
        For any $\Delta \in W_{can}$, the following holds:
        \begin{description}
            \item[1.] $\alpha \land \beta \in \Delta$ iff $\alpha,\beta \in \Delta$
            \item[2.]   $\alpha \lor \beta \in \Delta$ iff $\alpha \in \Delta$ or $\beta \in \Delta$
            \item[3.]  $\alpha \to \beta \in \Delta$ iff $\alpha \notin \Delta$ or $\beta \in \Delta$
            \item[4.]  if $\neg \alpha \notin \Delta$ then $\alpha \in \Delta$
            \item[5.]  If $\alpha \in \Delta$ and $\neg \alpha \in \Delta$ then $\circ \alpha \notin \Delta$
            \item[6.]  $\obl \alpha \in \Delta$ iff $\alpha \in \Delta\prim$ for all $\Delta\prim \in W$ such that $\Delta R_{can} \Delta\prim$.
        \end{description}
    \end{lemma}

    \begin{proof}
        Items 1. through 5. are immediate from the definitions and the fact that $\Delta$ is $\varphi$-saturated, hence, by \cref{prop:phi-sat}, it is closed under logical consequences. In particular, it contains any instance of the axioms of \DmbC, and it is closed under Modus Ponens.

        To prove the right-to-left direction for 6., assume that $\obl \alpha \notin \Delta$. Since $\Delta$ is closed under logical consequences, it follows that $\nvdash_{\DmbC} \alpha$, because of the $\obl$-necessitation rule. Suppose, for a contradiction, that $Den(\Delta) \vdash_{\DmbC} \alpha$. By \cref{def:derivations}, there are $\beta_1, \dots, \beta_n \in Den(\Delta)$ (for $n\geq 1$) such that  
        $$\vdash_{\DmbC} \beta \to \alpha,$$ 
        where $\beta=\beta_1 \land \ldots \land \beta_n$. By applying necessitation and \axK, 
        $$\vdash_{\DmbC} \obl \beta \to \obl \alpha.$$ 
        Observe now that $\obl\beta_i \in \Delta$, by definition of  $Den(\Delta)$, hence $\obl\beta_1 \land \ldots \land \obl \beta_n \in \Delta$, by item~1. But 
        $$\vdash_{\DmbC} (\obl\beta_1 \land \ldots \land \obl \beta_n) \to \obl\beta,$$ 
        hence $\obl \beta \in \Delta$. Using again that $\Delta$ is closed under logical consequences, we infer that  $\obl \alpha \in \Delta$, which contradicts our initial assumption. 

        We conclude, therefore, that $Den(\Delta) \nvdash_{\DmbC} \alpha$. But then, there is some $\Delta\prim$ such that $Den(\Delta) \subseteq \Delta\prim$ and $\Delta\prim$ is $\alpha$-saturated. Therefore, there is $\Delta\prim \in W_{can}$ such that $\Delta R_{can} \Delta\prim$ and $\alpha \notin \Delta\prim$.

        The proof of the left-to-right direction for 6. is immediate from the definitions. Indeed, if   $\obl \alpha \in \Delta$  and  $\Delta R_{can} \Delta\prim$ then $\alpha \in \Delta\prim$, given that $\alpha \in  Den(\Delta)$.
    \end{proof}
    
	\begin{prop}
		\label{prop:TruthLemma}
		The triple  $\M = \langle W_{can}, R_{can}, \{v_\Delta\}_{\Delta \in W_{can}} \rangle$ is a swap Kripke model for {\DmbC} such that $v_\Delta(\varphi) \in \mathcal{D}$ if and only if $\varphi \in \Delta$ if and only if $\alpha_{(1,\Delta)}=1$.
	\end{prop}
    \begin{proof}
		Observe first that $R_{can}$ is serial. To see this, let $\Delta \in W_{can}$. Then, $\Delta$ is $\varphi$-saturated, for some formula $\varphi$. Suppose that $\obl\alpha \in \Delta$, for every formula $\alpha$. Then, $\obl\alpha, \obl\neg\alpha, \obl{\circ}\alpha \in \Delta$ and so, since $\Delta$ is a closed theory, $\obl \bot_\alpha \in \Delta$. This shows that $\bot_\alpha \in \Delta$, by $(\obl-E)$ and so $\varphi \in \Delta$, a contradiction. From this, $\obl\alpha \notin\Delta$ for some $\alpha$. By reasoning as in the proof of item~6. of \cref{lemma:truth-sat}, we infer that $Den(\Delta) \nvdash_{\DmbC} \alpha$ and so there exists  some $\Delta\prim$ such that $Den(\Delta) \subseteq \Delta\prim$ and $\alpha \notin   \Delta\prim$. This shows that  $R_{can}$ is serial. The rest of the proof is an immediate consequence from \cref{prop:designated-dmbc} and \cref{def:can-bival-truth}.
	\end{proof}

    \begin{thm}[Completeness of \DmbC\ w.r.t. swap Kripke models] \ \\
		For any set $\Gamma \cup \{\varphi\} \subseteq For_\Sigma$, if $\Gamma \vDash_{\DmbC} \varphi$, then $\Gamma \vdash_{\DmbC} \varphi$.
	\end{thm}
	\begin{proof}
		Suppose, to the contrary, that $\Gamma \nvdash_{\DmbC} \varphi$. Thus, by \cref{prop:phi-sat}, there is some $\varphi$-saturated $\Delta \in W_{can}$ such that $\Gamma \subseteq \Delta$. Since $\Delta \nvdash_{\DmbC} \varphi$, then $\varphi \notin \Delta$. Let $\mathcal{M}^{K-S}_{DmbC}$ be the canonical swap model for DmbC. Then, $\mathcal{M}^{K-S}_{DmbC}, \Delta \vDash \Gamma$, since $\Gamma \subseteq \Delta$, but $\mathcal{M}^{K-S}_{DmbC}, \Delta \nvDash \varphi$.   This proves that $\Gamma \nvDash_{\DmbC} \varphi$.
	\end{proof}

\section{Some extensions of {\DmbC}}

\begin{defn}
	\label{def:ciw}
	Consider the following axioms over the signature $\Sigma$:
    \begin{description}
        \item[(ciw)] ${\circ}\alpha \lor (\alpha \land \neg \alpha)$ 
        \item[(ci)]  $\neg {\circ} \alpha \to (\alpha \land \neg \alpha)$
        \item[(cf)] $\neg \neg \alpha \to \alpha$
    \end{description}
    The following systems can be thus defined:
    \begin{itemize}
        \item ${\DmbCciw} \defeq \DmbC\cup\{\textbf{(ciw)}\}$ 
        \item ${\DmbCci} \defeq \DmbC\cup \{\textbf{(ci)}\}$
        \item ${\DbC} \defeq \DmbC\cup\{\textbf{(cf)}\}$
        \item ${\DCi} \defeq \DmbCci \cup \{\textbf{(cf)}\}$
    \end{itemize}
\end{defn}

\begin{rmk}
   This section will talk about results that can be easily adaptable to each of the aforementioned cases. From now on, by $\textbf{L}$ we will refer to any system in the set $\{\DmbCciw,\DmbCci,\DbC,\DCi\}$.
The notion of derivation $\Gamma \vdash_{\bf L}\varphi$ in {\bf L} is as in \cref{def:derivations} (with the corresponding set of axioms of each logic).
\end{rmk}

The multialgebras to accommodate each of the axioms are defined as follows:

\begin{defn} \label{def:swap-negbC}
    Let $\A_3$ be the swap structure for {\DmbC}. The multioperators of the swap structure for {\DbC} are defined as in \cref{def:swapOperations}, with the exception of $\aneg$, which is substituted for $$\aneg_1 a = \{c \in A \ : \  c_1 = a_2 \text{ and } c_2 \leq a_1 \}.$$
\end{defn}

The non-deterministic truth-table for $\aneg_1$ is as follows:

\begin{center}
    
\begin{tabular}{|l|c|}
\hline
 & $\aneg_1$ \\ \hline
$T$ & $\{F\}$ \\ \hline
$t$ & $\{T,t\}$\\ \hline
$F$ & $\{T\}$\\ \hline
\end{tabular}

\end{center}

\begin{defn}
    \label{def:swapOp-ext}
    Let $\A_3$ be the swap structure for {\DmbC}. The multioperators of the swap structure for {\DmbCciw}, {\DmbCci} and  {\DCi} are defined as in \cref{def:swapOperations}, with exception of the multioperators that are mentioned in this definition, which are substituted accordingly.

    \begin{enumerate}
        \item  In {\DmbCciw}: $\acirc_1 a := \{c \in A \ : \ c_1 = {\sim} (a_1 \sqcap a_2)\}$.
        \item  In {\DmbCci}: $\acirc_2 a = \{({\sim} (a_1 \sqcap a_2), a_1\sqcap a_2) \}$.
        \item In {\DCi}, take $\acirc_2$ from {\DmbCci} and $\aneg_1$ from {\DbC}.
    \end{enumerate}
\end{defn}

The non-deterministic truth-tables for $\acirc_1$ and $\acirc_2$ are as follows:

\begin{center}
\begin{tabular}{|l|c|}
\hline
 & $\acirc_1$ \\ \hline
$T$ & $\{T,t\}$ \\ \hline
$t$ & $\{F\}$\\ \hline
$F$ & $\{T,t\}$\\ \hline
\end{tabular}
\hspace{2cm}
\begin{tabular}{|l|c|}
\hline
 & $\acirc_2$ \\ \hline
$T$ & $\{T\}$ \\ \hline
$t$ & $\{F\}$\\ \hline
$F$ & $\{T\}$\\ \hline
\end{tabular}

\end{center}

\begin{defn}
\label{def:swapval-ext}
       For each $w \in W$, we establish the following:
       
       \begin{itemize}
           \item The swap valuations for {\DmbCciw}, $\nuciw_w$, are defined  as in \cref{def:nu_w} for all operators, except for $\circ$, which satisfies the following condition:$$\nuciw_w({\circ} \alpha) \in \acirc_1 \nuciw_w (\alpha). $$
           \item The swap valuations for {\DmbCci}, $\nuci_w$, are defined as  in \cref{def:nu_w} for all operators, except for $\circ$, which satisfies the following condition:$$\nuci_w({\circ} \alpha) \in \acirc_2 \nuci_w (\alpha). $$
           \item The swap valuations for {\DbC}, $\nubc_w$, are defined as  in \cref{def:nu_w} for all operators, except for $\neg$, which satisfies the following condition:$$\nubc_w({\neg} \alpha) \in \aneg_1 \nubc_w (\alpha). $$
           \item The swap valuations for {\DCi}, $\nuCi_w$, are  defined as  in \cref{def:nu_w} for all operators, except for $\neg$, which is defined using $\aneg_1$ as in the case of \DbC\ and for $\circ$, which is defined using $\acirc_2$ as in the case of {\DmbCci}.
       \end{itemize}
\end{defn}

\begin{defn}
    The structure $\M = \langle W, R, \{v_w^{\textbf{L}}\}_{w \in W} \rangle$ is a swap Kripke model for the logic $\textbf{L}$.
\end{defn}

We maintain an analogous notation to the one presented in \cref{rmk:notation}. Notice that this implies that $v_w^{\textbf{L}}(\alpha) \in \D$ if and only if $\alpha_{(1,w)} = 1$. We use this fact in the next proof.

\begin{lemma}
    \label{lemma:proj-nu-ext}
    Conditions 1.-6. listed on \cref{lemma:proj-nu_w} hold for $\textbf{L}$. Moreover, consider the following conditions:
    \begin{description}
        \item[5$\prim$.]  $(\circ \alpha)_{(1,w)} = {\sim} (\alpha_{(1,w)} \sqcap \alpha_{(2,w)})$.
        \item[5$^*$.] $(\circ\alpha)_{(1,w)} = {\sim} (\alpha_{(1,w)} \sqcap \alpha_{(2,w)})$ and $(\circ\alpha)_{(2,w)} = (\alpha_{(1,w)} \sqcap \alpha_{(2,w)}) $.
        \item[7.] $(\neg\alpha)_{(2,w)} \leq \alpha_{(1,w)}$.
    \end{description}
    Then, condition $5\prim$ holds in {\DmbCciw}; condition $5^*$ holds in {\DmbCci}; condition $7$ holds in {\DbC}; and conditions $5^*$ and $7$ hold in {\DCi}.
    %Conditions $5\prim$ and $5^*$ substitute condition $5$ in \cref{lemma:proj-nu_w}, while condition $7$ is added to the other ones.
\end{lemma}

\begin{proof}
 For {\DbC},  condition $7$ follows by  \cref{def:swap-negbC} and \cref{def:swapval-ext}: $\nubc_w (\neg\alpha) \in \aneg_1 \nubc_w (\alpha)$, and so $(\neg \alpha)_{(2,w)} \leq \alpha_{(1,w)}$.
 
    For {\DmbCciw} and  {\DmbCci}, condition $5\prim$ and $5^*$ follow, respectively, by \cref{def:swapOp-ext} and \cref{def:swapval-ext}. 

    The case of {\DCi} follows from  {\DmbCci} and {\DbC}. 
\end{proof}

\begin{defn}
		Let $\mathcal{M}^{K-S}_{\textbf{L}} = \langle W^{\textbf{L}}, R^{\textbf{L}}, \{\nu^{\textbf{L}}_w\}_{w \in W}  \rangle$ be as above. Then, a formula $\alpha \in For_\Sigma$ is said to be $\mathcal{M}^{K-S}_{\textbf{L}}$-true in a world $w$, denoted by $\mathcal{M}^{K-S}_{\textbf{L}}, w \vDash \alpha$, if it is the case that $\nu^{\textbf{L}}_w(\alpha) \in \mathcal{D}$.
	\end{defn}

	\begin{defn}
		Let $\Gamma \cup\{\alpha\} \subseteq For_\Sigma$. We say that $\alpha$ is a logical consequence of $\Gamma$ in {\bf L}, denoted by $\Gamma \vDash_{\bf L} \alpha$, if for all $\mathcal{M}^{K-S}_{\textbf{L}}$ and $w \in W$:  $\mathcal{M}^{K-S}_{\textbf{L}}, w \vDash \Gamma$ implies that $\mathcal{M}^{K-S}_{\textbf{L}}, w \vDash \alpha$.
	\end{defn}

    \begin{thm}[Soundness]
        For every $\Gamma \cup \{\varphi\} \subseteq For(\Sigma)$, $$\text{if }\Gamma \vdash_{\textbf{L}} \varphi \text{, then } \Gamma \vDash_{\textbf{L}} \varphi$$
    \end{thm}

    \begin{proof} Consider first \DmbCciw. Given a valuation $v_w$ and a formula $\alpha$, it must be shown that $(\circ\alpha \vee (\alpha \land \neg \alpha))_{(1,w)}=1$. By the  definition of the multioperator $\tilde{\vee}$ in \cref{def:swapOperations}, the latter is equivalent to prove that either $(\circ\alpha)_{(1,w)}=1$ or $(\alpha \land \neg \alpha)_{(1,w)}=1$. But this is immediate, by property $5\prim$ of $v_w$ given in \cref{lemma:proj-nu-ext} and the fact that $(\alpha \land \neg \alpha)_{(1,w)}=\alpha_{(1,w)} \sqcap \alpha_{(2,w)}$.

        In {\DmbCci}, it must be shown that \textbf{(ci)} is valid. Given $v_w$, assume that $(\neg{\circ}\alpha)_{(1,w)}=1$. Thus, we have $$(\neg{\circ}\alpha)_{(1,w)} = (\circ\alpha)_{(2,w)} = (\alpha_{(1,w)} \sqcap \alpha_{(2,w)}) =1,$$ 
    by property $5^*$ of $v_w$. This shows that $v_w$ satisfies any instance of axiom  \textbf{(ci)}.
    
       For \DbC, the case  for \textbf{(cf)} is immediate from  property $7$ in \cref{lemma:proj-nu-ext}.

       Clearly, soundness of {\DCi} follows from soundness of {\DmbCci} and {\DbC}.
    \end{proof}

    For every logic in $\textbf{L}$, we define $W^\textbf{L}_{can}, R^\textbf{L}_{can}$ and $\nu^\textbf{L}_{\Delta}$ following the definitions for {\DmbC} and adapting to each \textbf{L} accordingly. Notice that each $\Delta$ is now a $\psi$-saturated set in $W^\textbf{L}_{can}$ This allows us to state the following lemma:  

    \begin{lemma}
        \label{lemma:truth-sat-ext}
        For any $\Delta \in W^\textbf{L}_{can}$, all statements 1. through 6. in \cref{lemma:truth-sat} hold.
        
        For \DmbCciw, we have the following strengthening of statement 5.:
        \begin{description}
            \item[5$^+_1$.]  $\alpha \in \Delta$ and $\neg \alpha \in \Delta$ iff $\circ \alpha \notin \Delta$
        \end{description}

        For \DmbCci, we have an additional condition for $\circ$: 
        \begin{description}
            \item[5$^+_2$.]  If $\neg{\circ}\alpha \in \Delta$, then $\alpha \in \Delta$ and $\neg \alpha \in \Delta$.
        \end{description}

        For \DbC, we have an additional condition for $\neg$: 
        \begin{description}
            \item[5$^+_3$.]  If $\neg \neg \alpha \in \Delta$, then $\alpha \in \Delta$.
        \end{description}

        For \DCi, both conditions $5^+_2$ and $5^+_3$ are added.
    \end{lemma}

    \begin{proof}
       All conditions  are easily proved by using the respective new axiom of {\bf L}, and the fact that $\Delta$ is saturated (hence it is a closed theory). 
    \end{proof}

    \begin{prop}
		\label{prop:TruthLemma-ext}
		The triple  $\M_\textbf{L} = \{W^\textbf{L}_{can}, R^\textbf{L}_{can}, \{\nu^\textbf{L}_\Delta\}_{\Delta \in W_{can}}\}$ is a swap Kripke model for \textbf{L} such that $\nu^{\textbf{L}}_\Delta(\varphi) \in \mathcal{D}$ if and only if $\varphi \in \Delta$ if and only if $\alpha_{(1,\Delta)}=1$.
	\end{prop}

    \begin{proof}
		It is an immediate consequence from \cref{lemma:truth-sat-ext} and \cref{def:can-bival-truth}. For instance, in order to prove that $\nu^{\DmbCci}_\Delta(\circ\alpha) \in \acirc_2 \nu^{\DmbCci}_\Delta(\alpha)$, let $z:=\nu^{\DmbCci}_\Delta(\alpha)$. Suppose first that  $z \in \{T,F\}$. By \cref{def:can-bival-truth}, either $\alpha \notin \Delta$ or $\neg\alpha \notin \Delta$. By 5$^+_1$ and  5$^+_2$ of  \cref{lemma:truth-sat-ext}, $\circ\alpha  \in \Delta$ and $\neg {\circ}\alpha  \notin \Delta$. From this, $\nu^{\DmbCci}_\Delta(\circ\alpha)=T \in \{T\}= \acirc_2 z$. Now, if  $z =t$ then $\alpha,\neg\alpha \in \Delta$ and so, by 5$^+_1$, $\circ\alpha  \notin \Delta$. Hence, $\nu^{\DmbCci}_\Delta(\circ\alpha)=F \in \{F\}= \acirc_2 z$. In turn, in order to prove that  $\nu^{\DbC}_\Delta(\neg\alpha) \in \aneg_1 \nu^{\DbC}_\Delta(\alpha)$, let $z:=\nu^{\DbC}_\Delta(\alpha)$. If $z=T$ then $\alpha \in \Delta$ and $\neg\alpha \notin\Delta$. From this, $\nu^{\DbC}_\Delta(\neg\alpha)=F \in \{F\}= \aneg_1 z$.  If $z=t$ then $\alpha, \neg\alpha \in\Delta$. From this, $\nu^{\DbC}_\Delta(\neg\alpha) \in \{T,t\}= \aneg_1 z$. Finally, if  $z=F$ then $\alpha \notin \Delta$ and $\neg\alpha \in\Delta$. By 5$^+_3$ of  \cref{lemma:truth-sat-ext}, $\neg\neg\alpha  \notin \Delta$ and so $\nu^{\DbC}_\Delta(\neg\alpha)=T \in \{T\}= \aneg_1 z$. The other cases are treated analogously.
	\end{proof}

    \begin{thm}[Completeness]
		For any set $\Gamma \cup \{\varphi\} \subseteq For(\Sigma)$, if $\Gamma \vDash_\textbf{L} \varphi$, then $\Gamma \vdash_\textbf{L} \varphi$
	\end{thm}
	\begin{proof}
		This is proved in analogous way as to the {\DmbC} case.
	\end{proof}

\section{The da Costa axiom: the case of \DmbCcl}

In this section, as well as in Sections~\ref{sect:DCila}, \ref{sect:C1D} and~\ref{sect:DCn}, we will consider axiomatic extensions of  \DmbC\ which include, among others, the so-called {\em da Costa axiom}
$$\textbf{(cl)}  \ \  \neg (\alpha \land \neg \alpha) \to {\circ} \alpha $$
This move has strong consequences: as it was shown by Avron in~\cite[Theorem~11]{Avron2007a}, the logic  \mbCcl\ obtained by adding \textbf{(cl)} to \mbC, a well as other extensions of  \mbCcl\ (including \Cila, the version of da Costa's system $C_1$ in a signature with $\circ$) cannot be semantically characterized by a single finite $N$matrix. As shown by Coniglio and Toledo in~\cite{ConiglioToledo2022}, this issue can be overcome by considering, for each of such systems, a suitable (finite--valued)  Nmatrix and restrict the set  of permitted valuations by a (decidable) criterion, through the notion of {\em restricted Nmatrices}  (or RNmatrices).

We will adapt this technique to our swap Kripke models, in order to deal with the deontic expansion of (some of) such systems. 

This section starts our analysis by considering \DmbCcl. The logic {\DmbCcl} is the extension of \DmbC\ by adding axiom \textbf{(cl)}. It is easy to show that \DmbCcl\ is a proper extension of {\DmbCciw}: this follows from the fact that the system {\mbCcl} is a proper extension of the system {\mbCciw}, which is obtained from \mbC\ by adding \textbf{(ciw)} \cite[Corollary~3.3.30]{Carnielli2016}. 
We thus present a swap Kripke semantics for {\DmbCcl} as the corresponding one for  {\DmbCciw}, together with a restriction on their valuations.
%This approach closely follows the proposal given in~\cite{ConiglioToledo2022}, in which  $RN$matrices were introduced as a way to overcome the uncharacterizability of {\mbCcl} and da Costa systems $C_n$ by means of finite-valued $N$matrices.

%    In the sequel, we follow~\cite{ConiglioToledo2022} in restricting our valuations and show that the results shown there can also be applied to modal extensions of \mbCcl.

   \begin{defn}
        \label{def:nu_cL}
    A swap Kripke model $\M^{K-S}_{\DmbCcl} = \langle W,R,\{\nucL_w\}_{w \in W}\rangle$ for \DmbCcl\ is a swap Kripke model  for \DmbCciw\ such that each valuation $\nucL_w$ satisfies, in addition, the following condition:
    $$\text{If }\nucL_w(\alpha) = t \text{, \ then \ } \nucL_w(\alpha \land \neg \alpha) = T.$$ 
  \end{defn}

    \begin{defn}
		Let $\mathcal{M}^{K-S}_{DmbCcl} = \langle W, R, \{\nucL_w\}_{w \in W}  \rangle$ be a  swap Kripke model for \DmbCcl. We say that a formula $\alpha \in For(\Sigma)$ is $\mathcal{M}^{K-S}_{\DmbCcl}$-true in a world $w$, denoted by $\mathcal{M}^{K-S}_{\DmbCcl}, w \vDash \alpha$, if $\nucL_w(\alpha) \in \mathcal{D}$.
	\end{defn}
 
\begin{defn}
		Let $\Gamma \cup \{\alpha\ \subseteq For(\Sigma)$. We say that $\alpha$ is a logical consequence of $\Gamma$ in \DmbCcl, denoted by $\Gamma \vDash_{\DmbCcl} \alpha$, if for all $\mathcal{M}^{K-S}_{\DmbCcl}$ and all $w \in W$: $\mathcal{M}^{K-S}_{\DmbCcl}, w \vDash \Gamma$ implies that $\mathcal{M}^{K-S}_{\DmbCcl}, w \vDash \alpha$.
\end{defn}

    The following lemma will be useful for showing soundness of  \DmbCcl\ w.r.t. swap Kripke models semantics.

    \begin{lemma}
        \label{lemma:proj-nu-c}
        Given the notation on \cref{rmk:notation}, let $\nucL_w$ be a valuation in a swap Kripke model $\M^{K-S}_{\DmbCcl}$ for \DmbCcl. Then, the following holds, for every formula $\alpha$:
        \begin{description}
            \item[5$^{**}$.] $({\circ}\alpha)_{(1,w)} = (\alpha \land \neg \alpha)_{(2,w)}$
        \end{description}
        Hence, any instance of axiom \textbf{(cl)} is true in any world of any swap Kripke model for \DmbCcl.
    \end{lemma}

    \begin{proof}
        By \cref{lemma:proj-nu-ext},  $(\circ \alpha)_{(1,w)} = {\sim} (\alpha_{(1,w)} \sqcap \alpha_{(2,w)})$. Suppose that  $({\circ}\alpha)_{(1,w)} =1$. Then, $\alpha_{(1,w)} \sqcap \alpha_{(2,w)}=(\alpha \land \neg\alpha)_{(1,w)}=0$. By definition of $A$, it follows that  $(\alpha \land \neg\alpha)_{(2,w)}=1=({\circ}\alpha)_{(1,w)}$. Now, suppose that $({\circ}\alpha)_{(1,w)} = 0$. Then, $\alpha_{(1,w)} \sqcap \alpha_{(2,w)}=\alpha_{(1,w)} \sqcap (\neg\alpha)_{(1,w)}=(\alpha \land \neg\alpha)_{(1,w)}=1$. From this, $\alpha_{(1,w)} = \alpha_{(2,w)}=1$, which means that $\nucL_w(\alpha) = (1,1)=t$. By \cref{def:nu_cL}, $\nucL_w(\alpha \land \neg \alpha) =T= (1,0)$. Hence, $(\alpha \land \neg\alpha)_{(2,w)}=0=({\circ}\alpha)_{(1,w)}$.

        The latter shows that any instance of axiom \textbf{(cl)} is true in any world of any swap Kripke model for \DmbCcl.
    \end{proof}

\begin{thm}[Soundness of \DmbCcl\ w.r.t. swap Kripke models] \ \\
        For every $\Gamma \cup \{\varphi\} \subset For(\Sigma)$, $$\text{if }\Gamma \vdash_{\DmbCcl} \varphi \text{, then } \Gamma \vDash_{\DmbCcl} \varphi.$$
\end{thm}

    \begin{proof}
        It follows from soundness of $\DmbCciw$ and \cref{lemma:proj-nu-c}.
    \end{proof}

In order to prove completeness for {\DmbCcl}, we will use the canonical model once again, to be constructed as in the case for {\DmbCciw}. Since the restriction occurs only on the valuations, we take only those valuations that are restricted appropriately. Hence, $W^{\DmbCcl}_{can}, R^{\DmbCcl}_{can}$ and $\nu^{\DmbCcl}_{\Delta}$ are defined as in {\DmbCciw}. Now, each $\Delta \in W^{\DmbCcl}_{can}$ is a $\psi$-saturated set in ${\DmbCcl}$. Thus we have:

 \begin{lemma}
        \label{lemma:truth-sat-cl}
        For any $\Delta \in W^{\DmbCcl}_{can}$, all statements for {\DmbCciw} stated in \cref{lemma:truth-sat-ext} hold. Besides, we add the following statement:
        \begin{description}
            \item[7.] If $\neg (\alpha \land \neg \alpha) \in \Delta$, then $\circ\alpha \in \Delta$.
        \end{description}
    \end{lemma}
    \begin{proof}
        Given that \DmbCcl\ extends \DmbCciw, it is an immediate consequence of \cref{lemma:truth-sat-ext}, axiom \textbf{(cl)}, and the fact that $\Delta$ is a closed theory.
    \end{proof}

    \begin{prop}
		\label{prop:TruthLemma-cl}
		The triple  $\M = \{W^{\DmbCcl}_{can}, R^{\DmbCcl}_{can}, \{\nu^{\DmbCcl}_\Delta\}_{{\Delta \in W}^{\DmbCcl}_{can}}\}$, constructed as in the case of \DmbC, is a swap Kripke model for {\DmbCcl} such that $\nu^{{\DmbCcl}}_\Delta(\varphi) \in \mathcal{D}$ if and only if $\varphi \in \Delta$ if and only if $\alpha_{(1,\Delta)}=1$.
	\end{prop}
    \begin{proof}
      Observe that each valuation $\nu^{\DmbCcl}_\Delta$ is defined according to \cref{def:can-bival-truth}. Since  \DmbCcl\ extends \DmbCciw, it follows that $\M$ is a swap  Kripke model for {\DmbCciw} such that $\nu^{{\DmbCcl}}_\Delta(\varphi) \in \mathcal{D}$ if and only if $\varphi \in \Delta$ if and only if $\alpha_{(1,\Delta)}=1$. In order to show that each  $\nu^{\DmbCcl}_\Delta$  satisfies the additional condition of \cref{def:nu_cL}, suppose that  $\nu^{{\DmbCcl}}_\Delta(\alpha)=t$. By \cref{def:can-bival-truth}, $\alpha, \neg\alpha \in \Delta$ and so $\alpha \land \neg\alpha \in \Delta$. Suppose that $\neg(\alpha \land \neg\alpha) \in \Delta$. By \cref{lemma:truth-sat-cl}, $\circ\alpha \in \Delta$. But then, by axiom \textbf{(bc)}, $\beta \in \Delta$, for every formula $\beta$, a contradiction. From this, $\neg(\alpha \land \neg\alpha) \notin \Delta$, therefore $\nu^{{\DmbCcl}}_\Delta(\alpha\land \neg\alpha)=T$, by \cref{def:can-bival-truth}. This shows that $\M$ is, in fact, a swap Kripke model for \DmbCcl.      
	\end{proof}

From the previous results and the construction above, it is easy to show that completeness holds for \DmbCcl.

\begin{thm}[Completeness of \DmbCcl\ w.r.t. swap Kripke models] \ \\
	For any set $\Gamma \cup \{\varphi\} \subseteq For(\Sigma)$, if $\Gamma \vDash_{\DmbCcl} \varphi$, then $\Gamma \vdash_{\DmbCcl} \varphi$
\end{thm}

\section{Swap Kripke models for \DCila} \label{sect:DCila}

Consider the following axiom schemas for {\em consistency propagation}:
    \begin{description}
        \item[(ca$_{\#}$)] $({\circ} \alpha \land {\circ} \beta) \to {\circ}(\alpha \# \beta)$,  where $\# \in \{\land, \lor, \to\}$
    \end{description}
   
We now add to {\DmbC} the axioms {\bf (ci)}, {\bf (cl)}, {\bf (cf)} and {\bf (ca$_\#$)}, obtaining a logic called \DCila. Equivalently, \DCila\ is obtained from \DmbCcl\ by adding axioms  {\bf (ci)}, {\bf (cf)} and {\bf (ca$_\#$)}. The non-modal fragment of \DCila\ is called \Cila, and corresponds to da Costa logic $C_1$ presented over the signature with $\circ$  (see~\cite[Section~5.2]{Marcos2007}). Indeed,  \Cila\ is a conservative expansion of da Costa's ${C}_1$. As proved in~\cite[Theorem~11 and Corollary~6]{Avron2007a},  \Cila\ and $C_1$ are not characterizable by a single finite $N$matrix.

Based on the results presented in the previous section, as well as the characterization of \Cila\ in terms of a 3-valued $RN$matrix found in~\cite{ConiglioToledo2022}, in the sequel we will characterize  \DCila\ by means of swap Kripke models based on a suitable 3-valued  $RN$matrix for \Cila.

\begin{defn}
        \label{def:nu_Cila}
    A swap Kripke model $\M^{K-S}_{\DCila} = \langle W,R,\{\nucila_w\}_{w \in W}\rangle$ for \DCila\ is a swap Kripke model  for \DCi\ such that each valuation $\nucila_w$ satisfies, in addition, the following conditions, for $\# \in \{\land, \lor, \to\}$:
    $$\text{If }\nucila_w(\alpha) = t \text{, \ then \ } \nucila_w(\alpha \land \neg \alpha) = T$$
    $$\text{If } \nucila_w(\alpha), \nucila_w(\beta) \in \{T,F\} \text{, then } \nucila_w(\alpha \# \beta) \in \{T,F\}.$$   
  \end{defn}

Note that the first condition on  $\nucila_w$ coincides with the one for  $\nucL_w$.

The notions of satisfaction of a formula $\alpha$ in a  world $w$ of a  swap Kripke model $\mathcal{M}^{K-S}_{DCila}$ for \DCila, denoted by $\mathcal{M}^{K-S}_{\DCila}, w \vDash \alpha$, as well as the semantical consequence  of \DCila\ w.r.t.  swap Kripke models, denoted by $\vDash_{\DCila}$, are defined as in the previous cases.

The above definitions guarantee that the axioms   $\textbf{(cl)}$ and {\bf (ca$_{\#}$)}  hold. Indeed:

\begin{lemma}
        \label{lemma:proj-nu-cila}
        Any instance of the axioms   $\textbf{(cl)}$ and {\bf (ca$_{\#}$)} are true in any world of any swap Kripke model for \DCila.
    \end{lemma}
    \begin{proof} Concerning   $\textbf{(cl)}$, the result holds by \cref{lemma:proj-nu-c} and \cref{def:nu_Cila}. Fix now $\# \in \{\land, \lor, \to\}$ and $\nucila_w$. Observe that, for any $\alpha$, $\nucila_w(\alpha) \in \{T,F\}$ iff $\alpha_{(1,w)} \neq \alpha_{(2,w)}$ iff $\alpha_{(1,w)} \sqcap \alpha_{(2,w)}=0$  iff $(\circ \alpha)_{(1,w)}={\sim} (\alpha_{(1,w)} \sqcap \alpha_{(2,w)})=1$. From this, $\nucila_w({\circ} \alpha \land {\circ} \beta) \in  \mathcal{D}$ implies that $({\circ} \alpha \land {\circ} \beta)_{(1,w)}=1$, which implies that  $({\circ} \alpha)_{(1,w)}= ({\circ} \beta)_{(1,w)}=1$. As observed above, the latter implies that $\nucila_w(\alpha), \nucila_w(\beta) \in \{T,F\}$, and so $\nucila_w(\alpha \# \beta) \in \{T,F\}$, by \cref{def:nu_Cila}. But this implies that $({\circ} (\alpha \# \beta))_{(1,w)}=1$, that is, $\nucila_w({\circ}(\alpha \# \beta)) \in  \mathcal{D}$. This shows that any instance of axiom {\bf (ca$_{\#}$)} is true in any world of any swap Kripke model for \DCila.
    \end{proof}

    As shown in \cite{ConiglioToledo2022}, the above characterization of \Cila\ by means of a 3-valued $RN$matrix induces a decision procedure for this logic. Given that SDL is decidable (for instance, by tableaux systems), so is  its modal extension \DCila.
   
    Now,  soundness of  {\DCila} w.r.t. swap Kripke models  follows easily from the previous results.

    \begin{thm}[Soundness of \DCila\ w.r.t. swap Kripke models] \ \\
        For every $\Gamma \cup \{\varphi\} \subseteq For(\Sigma)$, if $\Gamma \vdash_{\DCila} \varphi$, then $\Gamma \vDash_{\DCila} \varphi$.
    \end{thm}

    \begin{proof}
        It follows from soundness of $\DCi$ and \cref{lemma:proj-nu-cila}.
    \end{proof}

    The proof of completeness is a straightforward adaptation of the case for $\DmbCcl$, by building the canonical model as in the case for {\DCi}, and by imposing suitable restrictions on the valuations. Thus, $W^{\DCila}_{can}, R^{\DCila}_{can}$ and $\nu^{\DCila}_{\Delta}$are defined as in {\DCi}, but now each $\Delta \in W^{\DCila}_{can}$ is a $\psi$-saturated set in ${\DCila}$. 

 \begin{lemma}
        \label{lemma:truth-sat-cila}
        For any $\Delta \in W^{\DCila}_{can}$, all statements for {\DCi} stated in \cref{lemma:truth-sat-ext} hold. Besides, $\Delta$ satisfies the following statements:
        \begin{description}
            \item[7.] If $\neg (\alpha \land \neg \alpha) \in \Delta$, then $\circ\alpha \in \Delta$.
            \item[8.] If ${\circ} \alpha, {\circ} \beta \in \Delta$, then ${\circ}(\alpha \# \beta) \in \Delta$,  where $\# \in \{\land, \lor, \to\}$.
        \end{description}
    \end{lemma}
    \begin{proof}
        \DmbCcl\ extends \DCi. From this, the result is an immediate consequence of \cref{lemma:truth-sat-ext}, axioms \textbf{(cl)} and  {\bf (ca$_\#$)}, and the fact that $\Delta$ is a closed theory.
    \end{proof}

    \begin{prop}
		\label{prop:TruthLemma-cila}
		The triple  $\M = \{W^{\DCila}_{can}, R^{\DCila}_{can}, \{\nu^{\DCila}_\Delta\}_{{\Delta \in W}^{\DCila}_{can}}\}$, constructed as in the case of \DmbC, is a swap Kripke model for {\DCila} such that $\nu^{{\DCila}}_\Delta(\varphi)
 \in \mathcal{D}$ if and only if $\varphi \in \Delta$ if and only if $\alpha_{(1,\Delta)}=1$.
	\end{prop}
    \begin{proof}
      Notice that each valuation $\nu^{\DCila}_\Delta$ is defined according to \cref{def:can-bival-truth}. Given that  \DCila\ extends \DCi, it follows that $\M$ is a swap  Kripke model for {\DCi} such that $\nu^{{\DCila}}_\Delta(\varphi) \in \mathcal{D}$ if and only if $\varphi \in \Delta$ if and only if $\alpha_{(1,\Delta)}=1$. Let us prove now that every  $\nu^{\DCila}_\Delta$  satisfies the additional conditions of \cref{def:nu_Cila}. The first condition is proved analogously to the case for \DmbCcl\ (see the proof of \cref{prop:TruthLemma-cl}). In order to prove the second condition  of \cref{def:nu_Cila}, observe first the following:\\[1mm]
      {\bf Fact:}  $\nu^{\DCila}_\Delta(\alpha) \in \{T,F\}$ if and only if  $\circ\alpha \in \Delta$.\\[1mm]
      Indeed, suppose first that  $\nu^{\DCila}_\Delta(\alpha) \in \{T,F\}$. By  \cref{def:can-bival-truth}, either $\alpha \notin \Delta$ or  $\neg\alpha \notin \Delta$. In both cases, $\alpha \land \neg \alpha \notin \Delta$, hence  $\neg(\alpha \land \neg \alpha) \in \Delta$. By axiom {\bf (cl)} and the properties of $\Delta$, $\circ\alpha \in \Delta$. Conversely, suppose that $\circ\alpha \in \Delta$.  By axiom {\bf (bc)} and the properties of $\Delta$, either $\alpha \notin \Delta$ or $\neg\alpha \notin \Delta$. By \cref{def:can-bival-truth}, $\nu^{\DCila}_\Delta(\alpha) \in \{T,F\}$.

      Fix now  $\# \in \{\land, \lor, \to\}$, and suppose that $ \nucila_\Delta(\alpha), \nucila_\Delta(\beta) \in \{T,F\}$. By the {\bf Fact}, ${\circ} \alpha, {\circ} \beta \in \Delta$. By axiom  {\bf (ca$_{\#}$)} and by taking into account that $\Delta$ is a closed theory, ${\circ}(\alpha \# \beta) \in \Delta$. By the {\bf Fact} once again, we infer that $\nucila_\Delta(\alpha \# \beta) \in \{T,F\}$.
      
      This shows that $\M$ is, in fact, a swap Kripke model for \DCila.      
	\end{proof}

From the previous results and the construction above, it is easy to show that completeness holds for \DCila.

    \begin{thm}[Completeness of \DCila\ w.r.t. swap Kripke models] \ \\
	For any set $\Gamma \cup \{\varphi\} \subseteq For(\Sigma)$, if $\Gamma \vDash_{\DCila} \varphi$, then $\Gamma \vdash_{\DCila} \varphi$.
    \end{thm}

    \section{The pioneering system ${C}^D_1$} \label{sect:C1D}

   We highlight the fact that {\Cila} is a conservative expansion of ${C}_1$ by adding the connective $\circ$. While {\Cila} has $\circ$ in its signature, ${C}_1$ defines consistency in terms of non-contradictoriness. That is to say, $C_1$ is defined over the  signature $\Sigma^{C_1} = \{\to , \neg, \lor, \land\}$ such that $\alpha^\circ := \neg (\alpha \land \neg \alpha)$, for any $\alpha \in For(\Sigma^{C_1})$. However, it is easy to show that if we substitute any appearance of $\circ \alpha$ in the axioms or rules for {\Cila} for $\alpha^\circ$, we get ${C}_1$. Moreover, the valid inferences in \Cila\  in the signature $\Sigma^{C_1}$ coincide with the ones in $C_1$ (see~\cite[Theorem~110]{Marcos2007}). Also noticeable is the treatment of strong negation in ${C}_1$, usually defined as $\cneg \alpha := \neg \alpha \land \alpha^\circ$. Let $\Sigma^{C_1}_D := \{\to , \neg, \lor, \land, \obl\}$. Because of the close relationship between  \Cila\ and $C_1$, the $\Sigma^{C_1}_D$-reduct of the swap Kripke models for \DCila\ characterize the deontic expansion \DCone\ of $C_1$, defined over $\Sigma^{C_1}_D$ by adding to $C_1$ the modal deontic axioms of \cref{def:DmbC},  but now by considering $\bot_\alpha:=(\alpha \land  \neg \alpha) \land \alpha^\circ$.
   Observe that, in $C_1$, axioms {\bf (bc)} and {\bf (ca$_{\#}$)} are now replaced  by the following:

	\begin{description}
		\item[(bc)'] $\alpha^\circ \to (\alpha \to (\neg \alpha \to \beta))$
		\item[(ca$_{\#}$)'] $(\alpha^\circ \land \beta^\circ) \to (\alpha \# \beta)^\circ$,  where $\# \in \{\land, \lor, \to\}$
	\end{description}

   In turn, axioms  {\bf (cl)} and {\bf (ci)} are not considered in $C_1$ (since they hold by the very definition of $(\cdot)^\circ$, as well as by axiom {\bf (cf)}).
   
   One other striking fact is that the pioneering paraconsistent deontic system  ${C}^D_1$ (also defined over $\Sigma^{C_1}_D$) proposed in~\cite{CostaCarnielli1986} has one more axiom in addition to the ones of \DCone, namely

	\begin{description}
		\item[(ca$_{\obl}$)'] $\alpha^\circ \to (\obl\alpha)^{\circ}$
	\end{description}

   Hence, we should be able to accommodate axiom {\bf  (ca$_{\obl}$)'} in our system. However, the system proposed so far does not account for it. Consider the following possible model of ${C}^D_1$. Each node $x$ shows a set $\Gamma$ as a label. This indicates that for every $\varphi \in \Gamma$, $v^{{C}_1}_{\Gamma}(\varphi) \in \D$. We will use an analogous notation in a few more examples: 
    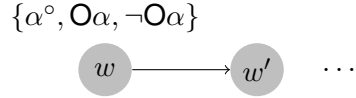
\begin{figure}[H]
    \centering
   \begin{tikzpicture}
       \begin{scope}%
        \node[vertex, label=above: {$\{\alpha^\circ,\obl \alpha, \neg \obl\alpha\}$}] (0) at (0,0) {$w$};
        \node[vertex] (2) at (2,0) {$w\prim$};
        \draw [->] (0) -- (2);
        \node (5) at (4,0) {};
        \path (2) -- (5) node [midway] {$\dots$};
  \end{scope}
   \end{tikzpicture}
   \caption{A representation of a counterexample to {\bf (ca$_{\obl}$)'}, when not adding the suitable restrictions to the models.}
    \end{figure}

   According to this model, $(\alpha^\circ)_{(1,w)} = 1$. But this does not say anything about any of the worlds $w\prim$ accessible to $w$. Notice that $v^{\DCone}_w((\obl \alpha)^{\circ}) \notin \D$ if and only if $v^{\DCone}_w(\obl \alpha) \in \D$ and $v^{\DCone}_w(\neg \obl \alpha) \in \D$. But this is perfectly possible, since when assigning a truth value for $\obl \alpha$, only the first coordinate of the snapshot is determined. Since the first coordinate of $v^{\DCone}_w(\neg \obl \alpha)$ is given by reading the second coordinate of $v^{\DCone}_w(\obl \alpha)$, then $v^{\DCone}_w(\obl \alpha) = (1,1)=t$ is the value that falsifies the formula correspondent to the axiom. Hence, in order to give a proper semantics for the original ${C}^D_1$, we need one more restriction. 
   
   From now on, for ease of notation, we will write $\nucone_w$ instead of $v^{C^D_1}_w$.%\add{Exemplo}

\begin{defn}
        \label{def:restriction-c1}
    A swap Kripke model $\M^{K-S}_{C^D_1} = \langle W,R,\{\nucone_w\}_{w \in W}\rangle$ for ${C}^D_1$ is a swap Kripke model  for \DCila\ (over the $\Sigma^{C_1}_D$-reduct) such that each valuation $\nucone_w$ satisfies, in addition, the following condition, for $\# \in \{\land, \lor, \to\}$:
%     $$\text{If } \nucone_w(\alpha), \nucone_w(\beta) \in \{T,F\} \text{, then } \nucone_w(\alpha \# \beta) \in \{T,F\};$$
     $$\text{If } \nucone_w(\alpha) \in \{T,F\} \text{, then } \nucone_w(\obl \alpha) \in \{T,F\}.$$ 
  \end{defn}

\begin{rmk} \label{rem:val-C1D}
    Observe that a family of maps $\nucone_w:For(\Sigma^{C_1}_D) \to A$ satisfies \cref{def:restriction-c1} iff it  satisfies the following conditions, for every $\alpha,\beta \in For(\Sigma^{C_1}_D)$ and $\# \in \{\land, \lor, \to\}$:
\begin{enumerate}
    \item items 1.-3. and 6. of \cref{def:nu_w};
    \item $\nucone_w({\neg} \alpha) \in \aneg_1 \nucone_w (\alpha)$, where $\aneg_1$ is as in \cref{def:swap-negbC};
    \item if $\nucone_w(\alpha) = t$, then $\nucone_w(\alpha \land \neg \alpha) = T$;  
    \item if  $\nucone_w(\alpha), \nucone_w(\beta) \in \{T,F\}$, then $\nucone_w(\alpha \# \beta) \in \{T,F\}$;
    \item if  $\nucone_w(\alpha) \in \{T,F\}$, then $\nucone_w(\obl \alpha) \in \{T,F\}$.
\end{enumerate}
\end{rmk}

    \begin{thm}[Soundness of ${C}^D_1$  w.r.t. swap Kripke models] \ \\
        For every $\Gamma \cup \{\varphi\} \subseteq For(\Sigma^{C_1}_D)$, if $\Gamma \vdash_{C^D_1} \varphi$, then $\Gamma \vDash_{C^D_1} \varphi.$
    \end{thm}  
    \begin{proof} Let $\nucone_w$ be as in \cref{def:restriction-c1}. Let us start by showing the following: \\[1mm]
   {\bf Fact:}  $\nucone_w(\alpha^\circ) \in \D$ if and only if   $\nucone_w(\alpha) \in \{T,F\}$.\\[1mm]
   Indeed, suppose that $\nucone_w(\alpha^\circ) \in \D$. Then,  $(\alpha^\circ)_{(1,w)}=1$. Recalling that $\alpha^\circ=\neg(\alpha \land \neg\alpha)$, it follows that $(\alpha \land \neg\alpha)_{(2,w)}=1$, and so  $\nucone_w(\alpha \land \neg \alpha) \neq T$. By the first condition in  \cref{def:nu_Cila} we infer that $\nucone_w(\alpha) \neq t$, hence $\nucone_w(\alpha) \in \{T,F\}$. Conversely, if $\nucone_w(\alpha) \in \{T,F\}$ then  $\alpha_{(1,w)} \sqcap \alpha_{(2,w)}=(\alpha \land \neg\alpha)_{(1,w)}=0$. From this, $(\alpha^\circ)_{(1,w)}=(\alpha \land \neg\alpha)_{(2,w)}=1$ and so $\nucone_w(\alpha^\circ) \in \D$.
   
   Now, assume that $\nucone_w(\alpha^\circ) \in \D$. By the {\bf Fact},  $\nucone_w(\alpha) \in \{T,F\}$. By the  condition in \cref{def:restriction-c1}, $\nucone_w(\obl \alpha) \in \{T,F\}$. Using the {\bf Fact} once again, we infer that  $\nucone_w((\obl\alpha)^\circ) \in \D$, and so axiom  {\bf  (ca$_{\obl}$)'} is valid w.r.t. swap Kripke models  for  ${C}^D_1$. 

   The validity of axiom {\bf (bc)'} follows immediately from the {\bf Fact}. In turn, the validity of  axiom {\bf (ca$_{\#}$)'} is a consequence of the first condition stated in \cref{def:nu_Cila} and  the {\bf Fact}.
   The validity of the other axioms of ${C}^D_1$ follows from the soundness of \DCila\ w.r.t. swap Kripke models.        
    \end{proof}

In order to prove completeness of ${C}^D_1$ w.r.t. swap Kripke models, some adaptations are required in the construction of the canonical  swap Kripke model and the canonical valuations.

First, observe that the $\varphi$-saturated sets in ${C}^D_1$ are subsets of $For(\Sigma^{C_1}_D)$ ($\circ$ is now a defined connective).

 \begin{lemma}
        \label{lemma:truth-sat-C1D} Let $\Delta \subseteq For(\Sigma^{C_1}_D)$ be a $\varphi$-saturated set in ${C}^D_1$. Then, it satisfies  the following properties, for every $\alpha,\beta \in For(\Sigma^{C_1}_D)$:
        \begin{description}
            \item[I.] items 1.-4. and 6. of \cref{lemma:truth-sat};
            \item[II.]  item 5$^+_3$. of \cref{lemma:truth-sat-ext};
            \item[III.]  $\alpha^\circ \in \Delta$  iff $\alpha \notin \Delta$ or $\neg\alpha \notin \Delta$;
            \item[IV.]  if $\alpha^\circ, \beta^\circ \in \Delta$, then $(\alpha \# \beta)^\circ \in \Delta$,  where $\# \in \{\land, \lor, \to\}$;
            \item[V.] if $\alpha^\circ \in \Delta$, then $(\obl\alpha)^\circ \in \Delta$.
        \end{description}
    \end{lemma}
\begin{proof}
    Items I. and II. are immediate, given that ${C}^D_1$ contains all the schemas of \DmbC\ and \DbC\ over signature $\Sigma^{C_1}_D$.\\[1mm]
    Item III.: The `only if' part is a consequence of axiom {\bf (bc)'}. Now, suppose that  $\alpha^\circ \notin \Delta$. By property 4. of item I., $\neg(\alpha^\circ) \in \Delta$. That is, $\neg\neg(\alpha \land \neg\alpha) \in \Delta$ and so $\alpha \land \neg\alpha \in \Delta$, by item II. By property 1. of item I., $\alpha, \neg\alpha \in \Delta$.\\[1mm]
    Item IV. and V. follow immediately from axioms {\bf (ca$_{\#}$)'} and  {\bf  (ca$_{\obl}$)'}.
\end{proof}

Define now $W^{{C}_1}_{can}, R^{{C}_1}_{can}$ and $\nucone_{\Delta}$  as in {\DCila}, but now each $\Delta \in W^{{C}_1}_{can}$ is a $\varphi$-saturated set in ${C}^D_1$. Observe that each valuation $\nucone_\Delta:For(\Sigma^{C_1}_D) \to A$ is defined according to \cref{def:can-bival-truth}.

    \begin{prop}
        The structure $\M_{{C}^D_1} = \langle W^{{C}_1}_{can}, R^{{C}_1}_{can}, \{v^{C_1}_\Delta\}_{\Delta \in W^{{C}_1}_{can}} \rangle$ is a swap Kripke model for ${C}^D_1$ such that, for every $\alpha \in For(\Sigma^{C_1}_D)$, $\nucone_\Delta(\alpha) \in \D$ iff $\alpha \in \Delta$.
    \end{prop}
\begin{proof} 
    
		It is an immediate consequence of \cref{lemma:truth-sat-C1D} and \cref{def:can-bival-truth}, taking into account \cref{rem:val-C1D}. Indeed, by adapting the proofs for the previous systems, it follows that $v^{C_1}_\Delta(\alpha \# \beta) \in v^{C_1}_\Delta(\alpha) \tilde{\#} v^{C_1}_\Delta(\beta)$ (for $\# \in \{\land,\vee,\to\}$) and $v^{C_1}_\Delta(\neg\alpha) \in \aneg_1 v^{C_1}_\Delta(\alpha)$. In order to prove that $v^{C_1}_\Delta$ satisfies the requirements 3.-5. of \cref{rem:val-C1D}, suppose first that  $v^{C_1}_\Delta(\alpha)=t$. Then, $\alpha,\neg\alpha \in \Delta$ and so $\alpha \land \neg\alpha \in \Delta$ and $\neg(\alpha \land \neg\alpha)=\alpha^\circ \not\in\Delta$, by items I. and III. of \cref{lemma:truth-sat-C1D}. This means that 
        $v^{C_1}_\Delta(\alpha \land \neg\alpha)=T$, validating requirement 3. For 4., observe first that  $v^{C_1}_\Delta(\alpha) \in \{T,F\}$ iff either  $\alpha\notin \Delta$ or  $\neg\alpha \notin \Delta$ iff, by III., $\alpha^\circ \in \Delta$. Now, let  $\# \in \{\land, \lor, \to\}$ and suppose that $\nucone_\Delta(\alpha), \nucone_\Delta(\beta) \in \{T,F\}$. By the previous observation, it follows that $\alpha^\circ, \beta^\circ \in \Delta$.  By IV. of \cref{lemma:truth-sat-C1D},  $(\alpha \# \beta)^\circ \in \Delta$. Using the observation above once again, this implies that $\nucone_\Delta(\alpha \# \beta) \in \{T,F\}$. The proof for 5{.} is analogous. By the very definitions, $\nucone_\Delta(\alpha) \in \D$ iff $\alpha \in \Delta$.
\end{proof}    

   Completeness follows immediately, with slight adaptations, from the lemma above and completeness for \DCila.

   \begin{prop}[Completeness of $C_1^D$  w.r.t. swap Kripke models] \ \\
	For any set $\Gamma \cup \{\varphi\} \subseteq For(\Sigma^{C_1}_D)$, if $\Gamma \vDash_{{C}^D_1} \varphi$ then $\Gamma \vdash_{{C}^D_1} \varphi$.
    \end{prop}

\section{Swap Kripke models for ${C}^D_n$} \label{sect:DCn}

We start this section by defining extensions of notions presented in the previous section  for the rest of the hierarchy ${C}_n$, for $n\geq 2$. For this, consider once again the signatures $\Sigma^{C_1}$ for $C_n$ and $\Sigma^{C_1}_D$ for the calculi $C^D_n$. We define the following  notation over $For(\Sigma^{C_1}_D)$:

\begin{itemize}
    \item $\alpha^0 = \alpha$
    \item $\alpha^{n+1} = \neg (\alpha^n \land \neg \alpha^{n})$
    \item $\alpha^{(n)}= \alpha^1 \land \dots \land \alpha^n$
\end{itemize}

We also follow the presentation of ${C}_n$ given in \cite{ConiglioToledo2022}. This comprises of all axioms for \cplp, plus (EM), (cf) and the following axioms:

\begin{description}
    \item[(bc$_n$)] $\alpha^{(n)} \to (\alpha \to  (\neg \alpha \to \beta))$ 
    \item[(P$_n$)] $(\alpha^{(n)} \land \beta^{(n)}) \to ((\alpha \to \beta)^{(n)} \land (\alpha \lor \beta)^{(n)} \land (\alpha \land \beta)^{(n)})$ 
\end{description}

Observe that the classical negation is represented in $C_n$ by means of the formula $\cneg^{(n)} \alpha := \neg \alpha \land \alpha^{(n)}$. From this, the new version of $\axD$ reads

\begin{description}
    \item[(D$_n$)] $\obl \alpha \to \cneg^{(n)} \obl \cneg^{(n)}\alpha$. 
\end{description}

We have to also highlight that, in order for a modal system to be constructed upon such logics in the way we have been presenting them, some decisions must be made. If we want to follow the presentation of $C^D_1$ on \cite{CostaCarnielli1986} and extend the ideas presented there, as we did in the previous section, we need to reformulate the classicality propagation axiom, namely, $\alpha^\circ \to (\obl \alpha)^\circ$, as follows: 

\begin{description}
    \item[(P$\obl_n$)] $\alpha^{(n)} \to (\obl \alpha)^{(n)}$. 
\end{description}

\noindent We call it the \textit{general classicality propagation} axiom in $C_n$. Notice that when $n = 1$, $\alpha^{(n)} = \alpha^\circ = \alpha^1 = \neg (\alpha \land \neg \alpha)$. Also notice that this has an influence on how strong a negation has to be in order to recover classicality. For $n = 1$, strong negation is already sufficient to introduce deontic explosion back into the system, but taking $n=2$, we have, besides $\neg$ and ${\sim}$, one more negation. We are in fact dealing with an increasing number of negations, or, more precisely, for each $n$, $C^{D}_n$ has $n+1$ negations.\footnote{Although the fact is easy to observe, the argument that each one of them is, in fact, a negation will be discussed in a future paper.}

We are now a position to characterize the family of systems  ${C}^D_n$, for $n \geq 1$. Also notice that the case where $n=1$ was already studied in the previous section. We also refrain in this from deeply investigating the philosophical considerations tied to these systems.  We opt for a technical development of a semantics for the systems proposed, with the general propagation of classicality and a distinct version of $\axD$. We attempt to maintain the general spirit of the system originally presentation in the paper by da Costa and Carnielli, while presenting a mix between RNmatrices and Kripke semantics.\footnote{It is possible to construct such a system by means of swap structures only, following the technique shown in \cite{ConPawSku2024}. The modal operator could then be assigned one dimension in the tuple, hence its truth value being fully nondeterministic. This permits to semantically characterize logics in which the modal operator does not satisfy any of the standard inference rules or axioms assumed for such an operator.}

We thus follow the presentation given in \cite{ConiglioToledo2022} to define the base system. So for each $n \geq 2$, the multialgebra for ${C}_n$ will have domain $A_n$ of size $n + 2$, where each element of $A_n \subseteq \textbf{2}^{n+1}$ is an ${n+1}$-tuple. Hence, the swap structures for ${C}_n$ is one where the set of \textit{snapshots} is: $$A_n = \{z \in \textbf{2}^{n+1} : (\bigwedge_{i \leq k} z_i) \lor z_{k + 1} = 1 \ \text{ for every } 1 \leq k \leq n \}$$  

This produces exactly the following $n+2$ truth values:
\begin{itemize}
    \item $T_n = (1, 0, 1, \dots, 1)$
    \item $t^n_0 = (1,1,0,1, \dots, 1)$
    \item $t^n_1 = (1,1,1,0,1,\dots, 1)$\\
    \vdots
    \item $t^n_{n-2} = (1,1,1,1,\dots,0)$
    \item $t^n_{n-1} = (1,1,1,1,\dots,1)$
    \item $F_n = (0,1,1,\dots,1)$.
\end{itemize}

\begin{defn}
    let $A_n$ be as in the definition above. We define the following subsets of $A_n$:
    \begin{description}
        \item[1.] $D_n \defeq A_n \setminus \{F_n\}$ (designated values)
        \item[2.] $U_n \defeq A_n \setminus D_n = \{F_n\}$  (undesignated values)
        \item[3.] $I_n \defeq A_n \setminus \{T_n, F_n\}$ (inconsistent values)
        \item[4.] $Boo_n = A_n \setminus I_n= \{T_n, F_n\}$ (Boolean or classical values)
    \end{description}
\end{defn}

Now we can introduce the multiagebra $\A_{{C}^D_n}$:

\begin{defn}
    \label{def:Cnop}
    Let $\A_{{C}^D_n} = (A_n, \aland, \alor, \ato, \aneg, \aobl)$ be the multialgebra over $\Sigma^{C_1}_D$ defined as follows, for any $a,b \in A_n$:
    \begin{description}
        \item[1.] $\aneg a = \{ c \in A_n \ : \ c_1 = a_2 \text{ and } c_2 \leq a_1\}$
        \item[2.] $a \aland b = \begin{dcases}
            \{c \in Boo_n \ : \ c_1 = a_1 \sqcap b_1\} & \text{ if } a, b \in Boo_n\\
            \{c \in A_n \ : \ c_1 = a_1 \sqcap b_1\} & \text{ otherwise } \\
        \end{dcases}$ %\{c \in A_n \ : \ c_1 = a_1 \sqcap b_1\}$
        \item[3.] $a \alor b = \begin{dcases}
            \{c \in Boo_n \ : \ c_1 = a_1 \sqcup b_1\} & \text{ if } a, b \in Boo_n\\
            \{c \in A_n \ : \ c_1 = a_1 \sqcup b_1\} & \text{ otherwise }\\
        \end{dcases}$
        \item[4.] $a \ato b = \begin{dcases}
            \{c \in Boo_n \ : \ c_1 = a_1 \supset b_1\} & \text{ if } a, b \in Boo_n\\
            \{c \in A_n \ : \ c_1 = a_1 \supset b_1\} & \text{ otherwise }\\
        \end{dcases}$
        \item[5.] $\aobl(X) = \{c \in A_n \ : \ c_1 = \bigsqcap \{x_1 : x \in X\}\}$, where $X \neq \emptyset$ and $X \subseteq A_n$.
    \end{description}
\end{defn}

\begin{rmk} 
\label{obs:tables-Cn}
Observe that the non-deterministic truth-tables for the non-modal operators of $\A_{{C}^D_n}$ are the ones displayed below, where $0 \leq i,j \leq n-1$.

\begin{center}
	\begin{tabular}{| c | c | c | c |}
		\hline $\tilde{\land}$ & $T_n$ & $t^n_j$ & $F_n$ \\
		\hline
		
		$T_n$ & $\{T_n\}$ &  $D_n$  &  $\{F_n\}$\\
		\hline
		
		$t^n_i$ &  $D_n$  & $D_n$  & $\{F_n\}$\\
		\hline		
		
		$F_n$ &  $\{F_n\}$  & $\{F_n\}$  & $\{F_n\}$\\
		\hline
		
	\end{tabular}
	\hspace{1cm}
	\begin{tabular}{| c | c | c | c |}
		\hline $\tilde{\lor}$ & $T_n$ & $t^n_j$ & $F_n$ \\
		\hline
		
		$T_n$ & $\{T_n\}$ &  $D_n$  &  $\{T_n\}$\\
		\hline
		
		$t^n_i$ &  $D_n$  & $D_n$  & $D_n$\\
		\hline
		
		$F_n$ &  $\{T_n\}$  & $D_n$  & $\{F_n\}$\\
		\hline
		
	\end{tabular}
\end{center}

\begin{center}
	\begin{tabular}{| c | c | }
		\hline  & $\tilde{\neg}$  \\
		\hline
		
		$T_n$ & $\{F_n\}$\\
		\hline
		
		$t^n_i$ &  $D_n$\\
		\hline
		
		$F_n$ &  $\{T_n\}$\\
		\hline
		
	\end{tabular}
	\hspace{3cm}
	\begin{tabular}{| c | c | c | c |}
		\hline $\tilde{\to}$ & $T_n$ & $t^n_j$ & $F_n$ \\
		\hline
		
		$T_n$ & $\{T_n\}$ &  $D_n$ &  $\{F_n\}$\\
		\hline
		
		$t^n_i$ &  $D_n$  & $D_n$  & $\{F_n\}$\\
		\hline
		
		$F_n$ &  $\{T_n\}$ & $D_n$  & $\{T_n\}$\\
		\hline
		
	\end{tabular}
\end{center}

\end{rmk}

\

\begin{defn}
    \label{def:Cnval}
    Let $W\neq \emptyset$ be a set of worlds, $R \subseteq W \times W$ be a serial relation and $v^n_w = For(\Sigma^{C_1}_D) \to A_n$ for each $w \in W$, such that, for any $\alpha, \beta \in For(\Sigma^{C_1}_D)$, the following holds:
    \begin{description}
        \item[1.] $v_w^n (\neg \alpha) \in \aneg (v_w^n (\alpha))$
        \item[2.] $v_w^n (\alpha \# \beta) \in v_w^n(\alpha) \tilde{\#} v_w^n (\beta)$, for $\# \in \{\land, \to \neg\}$
        \item[3.] $v_w^n (\obl \alpha) \in \aobl \big(\{v_{w\prim}^n(\alpha) \ : \ w R w\prim\}\big)$.
    \end{description}
\end{defn}

\begin{defn}
    \label{def:pre-swqp-Cn}
A  structure $\M = (W, R, \{v^n_w\}_{w \in W})$ with properties as in \cref{def:Cnval} is said to be a {\em swap Kripke pre-model} for ${C}^D_n$. A formula $\alpha \in For(\Sigma^{C_1}_D)$ is {\em true in a world $w$} of $\M$, denoted by $\M,w \vDash \alpha$, if $v^n_w(\alpha) \in D_n$. A formula $\alpha$  is {\em valid} in a pre-model $\M$, denoted by  $\M \vDash \alpha$, if  $\M,w \vDash \alpha$ for every $w \in W$. As it was done before, given a non-empty set $\Gamma$ of formulas we will write $\M,w \vDash \Gamma$ to denote that $\M,w \vDash \alpha$ for every $\alpha \in \Gamma$.
\end{defn}

We recall the fact that for any $\alpha \in For(\Sigma^{C_1}_D)$, $w \in W$ and $v^n_w$, $v^n_w(\alpha) \in D_n$ if and only if $v^n_{(1,w)}(\alpha) = 1$, given a natural adaptation of the notation presented in \cref{rmk:notation} and the definitions above.

In order to characterize ${C}^D_n$ we add first the following restrictions, thus simulating the behavior of the $RN$matrix for $C_n$:

\begin{defn}
    \label{def:Cnrest} Given a swap Kripke pre-model for ${C}^D_n$, consider the following additional restrictions on the valuations $v^n_w$: 
    \begin{description}
        \item[1.] $v^n_w(\alpha) = t^n_0$ implies $v^n_w(\alpha \land \neg \alpha) = T_n$
        \item[2.] $v_w^n(\alpha) = t^n_k$ implies $v^n_w(\alpha \land \neg \alpha) \in I_n$ and $v^n_w(\alpha^1)=t^n_{k-1}$,\\ for every $1\leq k \leq n-1$. 
    \end{description}

\end{defn}

\begin{rmk}
\label{val=sound-Cn}
We observe that the additional restrictions in \cref{def:Cnrest} only consider valuations in which the values of $\alpha^{(n)}$ are in $Boo_n$, such as shown in \cite[pp. 621--622]{ConiglioToledo2022}, for each world $w \in W$. Moreover, in any  swap Kripke pre-model  for ${C}^D_n$ as in \cref{def:Cnrest}, and for a fixed $w \in W$, each valuation $v^n_w$ belongs to the set of valuations of the $RN$matrix characterizing $C_n$ introduced in~\cite{ConiglioToledo2022}. From this, all the results concerning the non-modal operators of $C_n^D$ will hold w.r.t. the valuations of such a  swap Kripke pre-models.
\end{rmk}

The restrictions on the valuations made in \cref{def:Cnrest}  can be displayed by means of a very useful table, see~\cite[Table~1, p.~622]{ConiglioToledo2022}. By convenience of the reader, we reproduce in Figure~\ref{Table~1} a slightly expanded version of that table, which represents the possible scenarios concerning restricted valuations for $C_n$ (and so, for the non-modal fragment of ${C}^D_n$), according to  \cref{def:Cnrest}. In that table, $X^*$ means that the value $X$ is forced by a restriction on the corresponding valuation.

	\begin{figure}[h!]
		\centering
		\includegraphics[width=1.0\textwidth]{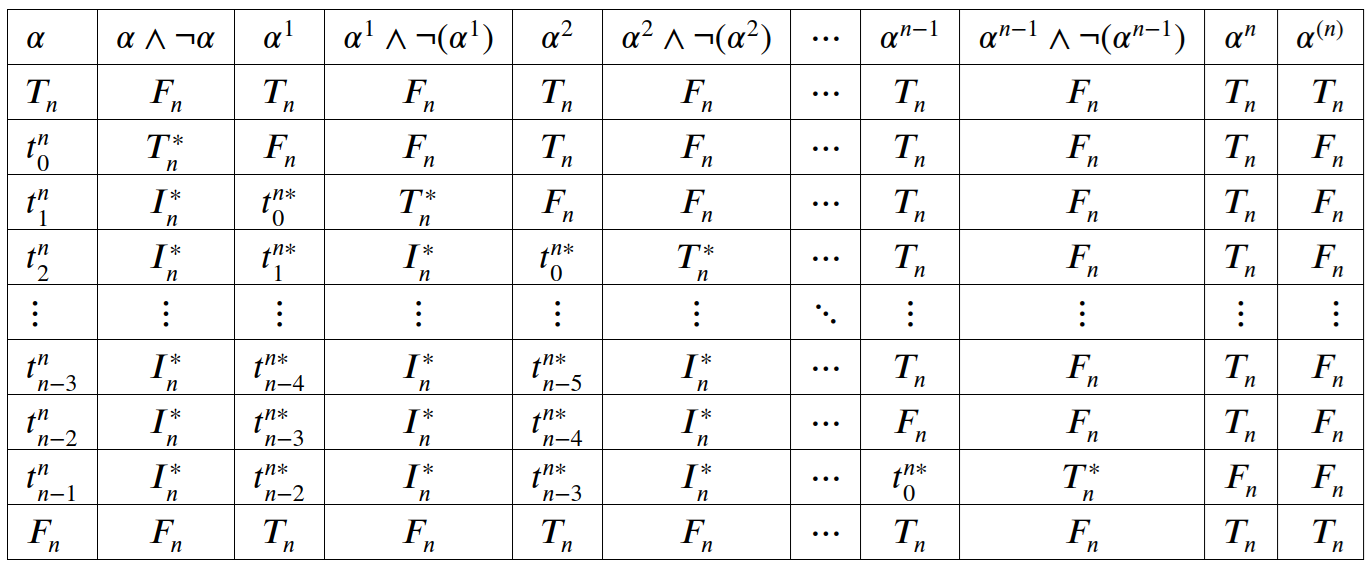}  \caption{}\label{Table~1}
	\end{figure}

In that case, it is worth observing that the truth-tables of the (defined) connectives $(\cdot)^{(n)}$ and   $\cneg^{(n)}\alpha=\neg\alpha \land \alpha^{(n)}$ are as follows, for $0 \leq i \leq n-1$:

\begin{center}
	\begin{tabular}{| c | c | }
		\hline $\alpha$ & $\alpha^{(n)}$  \\
		\hline
		
		$T_n$ & $\{T_n\}$\\
		\hline
		
		$t^n_i$ &  $\{F_n\}$\\
		\hline
		
		$F_n$ &  $\{T_n\}$\\
		\hline		
	\end{tabular}
\hspace{2cm}
	\begin{tabular}{| c | c | }
		\hline $\alpha$ & $\cneg^{(n)} \alpha$  \\
		\hline
		
		$T_n$ & $\{F_n\}$\\
		\hline
		
		$t^n_i$ &  $\{F_n\}$\\
		\hline
		
		$F_n$ &  $\{T_n\}$\\
		\hline		
	\end{tabular}
\end{center}

We need, however, to be sure that our restricted valuations preserve validity when looking at the modal operator. It is easy to see that axioms $\axK$ and the strong  version ({\bf D}$_n$) of $\axD$ are valid w.r.t.  the swap Kripke pre-models of \cref{def:Cnrest}:

\begin{lemma} 
\label{sound-axs-CnD}
Consider a swap Kripke pre-model $\M$ for ${C}^D_n$. Then, the following holds for any $\alpha, \beta \in For(\Sigma^{C_1}_D)$, and $w$ in $\M$:
   \begin{description}
       \item[1.] If $v^n_w(\obl (\alpha \to \beta)) \in D_n$ and $v^n_w(\obl \alpha) \in D_n$, then $v^n_w(\obl \beta) \in D_n$.
       \item[2.] If $v^n_w(\obl \alpha) \in D_n$, then $v^n_w(\cneg^{(n)} \obl \cneg^{(n)} \alpha) \in D_n$, assuming that $\M$ is  as in \cref{def:Cnrest}.
   \end{description}
\end{lemma}
\begin{proof}
    For the first item, assume both $v^n_w(\obl (\alpha \to \beta)) \in D_n$ and $v^n_w(\obl \alpha) \in D_n$. This means that $v^n_{(1,w)}(\obl (\alpha \to \beta))=v^n_{(1,w)}(\obl \alpha)=1$. But then, by \cref{def:Cnop} and \cref{def:Cnval} we have that $v^n_{(1,w')}(\alpha \to \beta)=v^n_{(1,w')}(\alpha)=1$ and so $v^n_{(1,w')}(\beta)=1$, for every $w'$ such that $wRw'$. From this, $v^n_{(1,w)}(\obl \beta)=1$, that is,  $v^n_w(\obl \beta) \in D_n$. 

    The second item follows from the nature of $\cneg^{(n)}$ (see its truth-tables above).       
\end{proof}

\begin{rmk}
    Notice that, since $n \geq 2$, the restrictions in \cref{def:Cnrest}  fail to validate $\axD$ in its strong negation form, i.e., w.r.t. $\cneg\alpha :=\neg\alpha \land \alpha^1$. 
    
    As an example, consider ${C}^D_2$. Picture a model as below:

    \begin{figure}[H]
    \centering
   \begin{tikzpicture}
       \begin{scope}
        \node[vertex] (0) at (0,0) {$w$};
        \node[vertex] (2) at (2,0) {$w\prim$};
        \node (3) at (4,0) {};
        \draw [->] (0) -- (2);
        \draw [->] (2) -- (3);
        \node (5) at (5,0) {};
        \path (3) -- (5) node [midway] {$\dots$};
  \end{scope}
   \end{tikzpicture}
    \end{figure}

    Assume that $v^2_w(\obl\alpha) \in D_2$ and that $v^2_{w\prim}(\alpha) = t^2_1$. Thus, by \cref{def:Cnrest}, this means that $v^2_{w\prim}(\alpha^1) = t^2_0$ and that $v^2_{w\prim}(\neg\alpha) \in D_2$. Hence, $v^2_{(1,w\prim)}(\cneg\alpha) = 1$. That implies $v^2_w(\obl \cneg \alpha) \in D_2$, so when $v^2_w( \obl \cneg \alpha) = T_2$, $v^2_w(\cneg \obl \cneg \alpha) = F_2$. The argument for any $n \geq 2$ is similar. Besides, a similar counterexample can be found for the paraconsistent negation $\neg$.
\end{rmk}

\begin{rmk}
    Also notice that the version of $\axD$ with $\cneg^{(n)}$ is validated by requiring the restrictions in \cref{def:Cnrest}. Without the restrictions, it is possible to construct a model and a valuation in worlds that would falsify the axiom. Picture again the model used in the previous remark, while working in $C^D_2$. Assume further that $v^2_w(\obl\alpha) \in D_2$, with $v^2_{w\prim}(\alpha) \in I_2$. This means that $v^2_{w\prim}(\neg\alpha) \in D_2$, hence $v^2_{w\prim}(\alpha \land \neg\alpha) \in D_2$. Since no restriction is given to the value assigned to $(\alpha \land \neg \alpha)$, then it is possible that $v^2_{w\prim}(\alpha \land \neg\alpha) \in I_2$, and also that $v^2_{w\prim}(\neg (\alpha \land \neg\alpha)) = v^2_{w\prim}(\alpha^1) \in I_2$. These assignments allow for $v^2_{w\prim}(\neg (\neg (\alpha \land \neg\alpha)\land \neg \neg (\alpha \land \neg\alpha))) = v^2_{w\prim}(\alpha^2) \in D_2$. But then, $v^2_{w\prim}(\cneg^{(2)} \alpha) \in D_2$. As we did in the previous case, taking $v^2_w(\obl\cneg^{(2)}\alpha) = T_2$ implies $v^2_w(\cneg^{(2)}\obl\cneg^{(2)}\alpha) = F_2$. This is also similarly extended to any $n \geq 2$.

    To see that indeed the restrictions guarantee that the axiom holds, assume $v^n_w(\obl \alpha) \in D_n$. Thus for all $w\prim$ such that $w R w\prim$, $v_{w\prim}^n(\alpha) \in D_n$. Now either $v_{w\prim}^n(\alpha) = T_n$ or $v_{w\prim}^n(\alpha) \in I_n$. If the first case, then $v_{w\prim}^n(\neg \alpha) = F_n$. Otherwise, $v_{w\prim}^n(\alpha^{(n)}) = F_n$. In any case, $v_{w\prim}^n(\cneg^{(n)}\alpha) = F_n$, hence $v^n_w(\obl \cneg^{(n)} \alpha) = F_n$, thus $v^n_w(\cneg^{(n)} \obl \cneg^{(n)} \alpha) \in D_n$. 
\end{rmk}

 The following is easily proved, taking into consideration \cref{val=sound-Cn} and  Figure~\ref{Table~1}:
 
\begin{lemma}
    \label{lemma:CnRN}
     Consider a swap Kripke pre-model  for ${C}^D_n$ as in \cref{def:Cnrest}, and let $1 \leq k \leq n$. Then, the following holds for any $\alpha \in For(\Sigma^{C_1}_D)$ and any $w \in W$:
    \begin{description}
        \item[1.] If $v^n_w(\alpha) = T_n$, then $v^n_w(\alpha^k) = T_n$
        \item[2.] If $v^n_w(\alpha)= t^n_i$ for some $1 \leq i \leq k - 2$, then $v^n_w(\alpha^k)=T_n$
        \item[3.] If $v^n_w(\alpha)= t^n_{k-1}$, then $v^n_w(\alpha^k)=F_n$
        \item[4.] If $v^n_w(\alpha)= t^n_i$ for some $k \leq i \leq n - 1$, then $v^n_w(\alpha^k)=t^n_{i - k}$
        \item[5.] If $v^n_w(\alpha) = F_n$, then $v^n_w(\alpha^k) = T_n$
        \item[6.] If $v^n_w(\alpha)= t^n_i$ for some $0 \leq i \leq n - 1$, then: $v^n_w(\alpha^k)=F_n$ iff $k=i+1$.
    \end{description}
\end{lemma}

\begin{lemma}
    \label{lemma:CnBoo}
    Let $\M_n$ be a swap Kripke pre-model  for ${C}^D_n$ as in \cref{def:Cnrest}. Then, $v^n_w(\alpha^{(n)}) \in D_n$ if and only if $v^n_w(\alpha) \in Boo_n$.
\end{lemma}
\begin{proof}
    Suppose $v^n_w(\alpha) \notin Boo_n$. Then, $v^n_w(\alpha) \in I_n$. By item~3. of \cref{lemma:CnRN} it follows that, for $0 \leq k < n$, $v^n_w(\alpha^{k + 1}) = F_n$, hence $v_w^n(\alpha^{(n)}) = F_n$ and so  $v^n_w(\alpha^{(n)}) \notin D_n$. Now, if $v^n_w(\alpha) \in Boo_n$ then, for any $1 \leq k \leq n$, $v^n_w(\alpha^k) = T_n$, by items~1. and~5. of \cref{lemma:CnRN}. Thus, $v^n_w(\alpha^{(n)}) = T_n$, so $v^n_w(\alpha^{(n)}) \in D_n$.   
\end{proof}

There is only one more modal axiom to check, namely, ({\bf P$\obl$}$_n$). In its original formulation, that is, when $n = 1$, it was already covered in the previous section. We thus add a similar restriction in order to validate this version of the axiom:

\begin{defn}
   \label{def:Cnrest3} A swap Kripke pre-model  for ${C}^D_n$ is said to be a  {\em swap Kripke model  for ${C}^D_n$} if the valuations satisfy, in addition, the restrictions of  \cref{def:Cnrest} plus the following constraint:
    $$\text{If } v_w^n(\alpha) \in Boo_n \text{, then } v_w^n(\obl \alpha) \in Boo_n$$
\end{defn}

Now we can prove the validity of ({\bf P$\obl$}$_n$) w.r.t. swap Kripke models for  ${C}^D_n$.

\begin{lemma}
\label{lemma:Cnax3}
    The following holds in any swap Kripke model for  ${C}^D_n$:
\begin{description}
    \item[3.] If $v_w^n(\alpha^{(n)}) \in D_n$, then $v_w^n((\obl \alpha)^{(n)}) \in D_n$.
\end{description}
\end{lemma}
\begin{proof}
    From \cref{lemma:CnBoo}, $v_w^n(\alpha^{(n)}) \in D_n$ implies that  $v_w^n(\alpha) \in Boo_n$. By \cref{def:Cnrest3}, $v_w^n(\obl \alpha) \in Boo_n$. By \cref{lemma:CnBoo} once again, $v_w^n((\obl \alpha)^{(n)}) \in D_n$. 
\end{proof}

Recall the notions and notation introduced in \cref{def:pre-swqp-Cn}, which can be also applied to swap Kripke models for  ${C}^D_n$.

\begin{defn}
     Given a set $\Gamma \subseteq For(\Sigma^{C_1}_D)$, we say that $\alpha$ is a logical consequence of $\Gamma$ in ${C}^D_n$, denoted by $\Gamma \vDash_{C^D_n} \alpha$, if the following holds: for every swap Kripke model $\M$ for  ${C}^D_n$, and for every world $w$ in $\M$, if $\M,w \vDash \Gamma$ then $\M,w \vDash \alpha$.
\end{defn}

\begin{thm}[Soundness of $C^D_n$ w.r.t. swap Kripke models] \ \\
    Let $\Gamma \cup \{\varphi\} \subseteq For(\Sigma^{C_1}_D)$. Then: $\Gamma \vdash_{C^D_n} \varphi$  only if $\Gamma \vDash_{C^D_n} \varphi$.   
\end{thm}
\begin{proof}
    The validity of the propositional (non-modal) axioms was already proven in~\cite{ConiglioToledo2022}. Since our construction is similar to that, we simply refer to the proof thus given, taking into consideration \cref{val=sound-Cn}. The cases for the modal axioms follow from lemmas~\ref{sound-axs-CnD} and~\ref{lemma:Cnax3}. Clearly, $\obl$-necessitation preserves validity, by item~6. of \cref{def:nu_w}.
\end{proof}

In order to prove completeness, on the other hand, we need a canonical construction that satisfies our new restrictions.
Let $W_{can}^{(n)}$ be the set of all the sets $\Delta \subseteq For(\Sigma^{C_1}_D)$ such that $\Delta$ is a  $\psi$-saturated sets in $C^D_n$, for some $\psi \in For(\Sigma^{C_1}_D)$. The binary relation $R_{can}^{(n)}$ on $W_{can}^{(n)}$ is defined as in the previous cases. Then:

\begin{lemma} [Truth Lemma for  $C^D_n$]
        \label{lemma:Cntruth}
        For any $\Delta \in W_{can}^{(n)}$, all the following statements hold, for every $\alpha,\beta \in For(\Sigma^{C_1}_D)$:
        \begin{description}
            \item[1.]  $\alpha \land \beta \in \Delta$ iff $\alpha,\beta \in \Delta$
            \item[2.]  $\alpha \lor \beta \in \Delta$ iff $\alpha \in \Delta$ or $\beta \in \Delta$
            \item[3.]  $\alpha \to \beta \in \Delta$ iff $\alpha \notin \Delta$ or $\beta \in \Delta$
            \item[4.]  If $\alpha \notin \Delta$, then $\neg \alpha \in \Delta$
            \item[5.]  If $\neg\neg\alpha \in \Delta$,  then $\alpha \in \Delta$
            \item[6.]  If  $\alpha \notin \Delta$ or $\neg \alpha \notin \Delta$, then $\alpha^1 \in \Delta$ and  $\neg(\alpha^1) \notin \Delta$
            \item[7.]  If  $\alpha, \neg \alpha \in \Delta$ then, for every $1 \leq i \leq n$: if $\alpha^i \notin \Delta$, then $\alpha^j \in \Delta$ for every $1 \leq j \leq n$ with $j \neq i$ 
            \item[8.]  If $\alpha, \neg\alpha \in \Delta$ then there exists a unique $1 \leq k \leq n$ such that $\alpha^{k} \notin \Delta$
            \item[9.]  $\alpha^{(n)} \in \Delta$ iff   $\alpha \notin \Delta$ or $\neg \alpha \notin \Delta$
            \item[10.]  $\obl \alpha \in \Delta$ iff $\alpha \in \Delta\prim$ for all $\Delta\prim \in W_{can}^{(n)}$ such that $\Delta R_{can}^{(n)} \Delta\prim$
            \item[11.]  If $\alpha^{(n)} \in \Delta$, then $(\obl\alpha)^{(n)} \in \Delta$
        \end{description}
    \end{lemma}
        \begin{proof}
        Conditions 1.-5. and 10.-11. are proven as in the previous cases, taken into account the axioms and rules of  $C^D_n$.\\[1mm]
        6.: Suppose that  $\alpha \notin \Delta$ or $\neg \alpha \notin \Delta$. By item~1., $\alpha \land \neg\alpha \notin \Delta$ and so $\alpha^1=\neg(\alpha \land \neg\alpha) \in \Delta$, by item~4. Since $\alpha \land \neg\alpha \notin \Delta$ then, by item~5., $\neg(\alpha^1) = \neg\neg(\alpha \land \neg\alpha) \notin \Delta$.\\[1mm]
        7.: Observe that item~6. is equivalent to the following:
        $$(*) \ \ \ \ \mbox{If} \ \ \alpha^1 \notin \Delta \ \mbox{ or } \ \neg(\alpha^1) \in \Delta \ \mbox{ then } \ \alpha, \neg \alpha \in \Delta.$$       
        By induction on $1 \leq i \leq n$ it will be proven that \\[2mm]
\indent        
$\begin{array}{ll}
   P(i):=  & \mbox{ for every $\alpha$, if $\alpha, \neg \alpha \in \Delta$ and $\alpha^i \notin \Delta$, then $\alpha^j \in \Delta$} \\
     & \mbox{  for every $1 \leq j \leq n$ with $j \neq i$}
\end{array}$\\[2mm]
        holds, for every $1 \leq i \leq n$ (for a given $n \geq 2$).\\[1mm]
        Base $i=1$: Assume that $\alpha, \neg \alpha \in \Delta$ and $\alpha^1 \notin \Delta$. By item~6. (applied to $\alpha^1$) it follows that  $\alpha^2 \in \Delta$ and  $\neg(\alpha^2) \notin \Delta$. By applying iteratively the same reasoning,  we infer that  $\alpha^j \in \Delta$ for every $1 \leq j \leq n$ with $j \neq 1$. That is, $P(1)$ holds.\\[1mm] 
        Inductive step: Assume that $P(i)$ holds for every $1 \leq i \leq k \leq n-1$, for a given $1 \leq k \leq n-1$ (Inductive Hypothesis, IH). Let $\alpha, \neg \alpha \in \Delta$ and suppose that $\alpha^{k+1}=(\alpha^k)^1 \notin \Delta$. By item~4.,  $\neg((\alpha^k)^1) \in \Delta$ and so $\alpha^k, \neg (\alpha^k) \in \Delta$, by $(*)$. Since $(\alpha^k)^1 \notin \Delta$ then, by (IH) applied to $\alpha^k$, it follows that $(\alpha^k)^j \in \Delta$ for every $1 \leq j \leq n$ with $j \neq 1$. Since $\alpha^k \in \Delta$, this implies that $$(i) \ \ \ \alpha^j \in \Delta \ \mbox{ for every $k \leq j \leq n$ with $j \neq k+1$.}$$ 
        In turn, since  $\neg((\alpha^{k-1})^1) = \neg (\alpha^k) \in \Delta$, then $\alpha^{k-1}, \neg (\alpha^{k-1}) \in \Delta$, by $(*)$. By applying iteratively the same reasoning,  we infer that  
        $$(ii) \ \ \ \alpha^j \in \Delta \ \mbox{ for every $1 \leq j \leq k-1$.}$$ 
        From (i) and (ii) it follows that $\alpha^j \in \Delta$, for every $1 \leq j \leq n$ with $j \neq k+1$. That is, $P(k+1)$ holds.\\[1mm]
        8.: It is an immediate consequence of item~7.\\[1mm]
        9.: The `Only if' part is immediate, by axiom ({\bf bc}$_n$) and the fact that $\Delta$ is a closed, non-trivial theory. Now, assume that $\alpha \notin \Delta$ or $\neg \alpha \notin \Delta$. By item~6., $\alpha^1 \in \Delta$ and  $\neg(\alpha^1) \notin \Delta$. By  item~6. applied to $\alpha^1$, and taking into account that   $\neg(\alpha^1) \notin \Delta$, it follows that $\alpha^2 \in \Delta$ and  $\neg(\alpha^2) \notin \Delta$. By applying iteratively the same reasoning,  we infer that  $\alpha^j \in \Delta$ for every $1 \leq j \leq n$, hence $\alpha^{(n)} \in \Delta$, by item~1.
    \end{proof}

\begin{defn}
    \label{def:Cn-bival-truth}
		For each $\Delta \in W_{can}^{(n)}$, define $\nu_\Delta^n: For(\Sigma^{C_1}_D) \to \A_n$ such that for each $\Delta \in W_{can}^{(n)}$ we have:
        \begin{equation*}
            v_\Delta^n(\alpha) = 
            \begin{dcases}
                T_n,& \text{if } \alpha \in \Delta, \neg\alpha \notin \Delta\\[1mm]
                t^n_k, & \text{if } \alpha, \neg\alpha \in \Delta \text{ and } \alpha^{k + 1} \notin \Delta\\[1mm]
                F_n,& \text{if } \alpha \notin \Delta, \neg\alpha \in \Delta 
            \end{dcases}
        \end{equation*}  
	\end{defn}

\begin{cor} \label{coro:v_Delta-CnD}
Let  $\Delta \in W_{can}^{(n)}$. Then, the following holds:\\[1mm]
1. The function $ v_\Delta^n$ is well-defined.\\[1mm]
2. $v^n_\Delta(\alpha) \in \{T_n,F_n\}$ iff $\alpha \notin\Delta$ or  $\neg\alpha \notin\Delta$, iff $\alpha^{(n)} \in \Delta$.\\[1mm]
3. $v^n_\Delta(\alpha) =t^n_i$ iff $\alpha^{i+1} \notin \Delta$.
\end{cor}
\begin{proof} Item 1. is an immediate consequence of item~8. of \cref{lemma:Cntruth}. In turn, item 2. follows by  item~9. of \cref{lemma:Cntruth} and the definition of $v^n_\Delta$. Finally, item~3. is a consequence of  item 1. and the definition of $v^n_\Delta$.
\end{proof}

Consider now the relation $R^{(n)}_{can} \subseteq  W_{can}^{(n)} \times  W_{can}^{(n)}$ defined as in the previous cases.

    \begin{prop} 
    \label{prop:can-model-CnD}
        The structure $\M_{n} = \langle W^{(n)}_{can}, R^{(n)}_{can}, \{v^{n}_\Delta\}_{\Delta \in W^{(n)}_{can}} \rangle$ is a swap Kripke model for ${C}^D_n$ such that, for every $\alpha \in For(\Sigma^{C_1}_D)$, $v^n_\Delta(\alpha) \in D_n$ iff $\alpha \in \Delta$.
    \end{prop}
\begin{proof}
Observe first that, by \cref{def:Cn-bival-truth},  for every $\Delta$ and every $\alpha$ it holds:
$$(*) \hspace{1cm} v^n_\Delta(\alpha) \in D_n  \ \mbox{ iff } \ \alpha \in \Delta.$$
(I) Let us prove now that each function $v^n_\Delta$ satisfies the properties stated in \cref{def:Cnval}. Concerning conjunction, observe that, by $(*)$ and item~1. of \cref{lemma:Cntruth}, $v^n_\Delta(\alpha \land \beta) \in D_n$ iff $v^n_\Delta(\alpha), v^n_\Delta(\beta) \in D_n$. In turn,  $z,w \in D_n$ implies that $z \tilde{\land} w \subseteq D_n$, and $z=F_n$ or $w=F_n$ implies that $z \tilde{\land} w =\{F_n\}$. Moreover, by item~2. of \cref{coro:v_Delta-CnD}: $v^n_\Delta(\alpha), v^n_\Delta(\beta) \in Boo_n$ implies that  $\alpha^{(n)}, \beta^{(n)} \in \Delta$ and so  $(\alpha \land \beta)^{(n)} \in \Delta$, by ({\bf P}$_n$), then $v^n_\Delta(\alpha \land\beta) \in Boo_n$. Given that $z,w \in Boo_n$ implies that  $z \tilde{\land} w \subseteq Boo_n$, we infer from the previous considerations that $v^n_\Delta(\alpha \land \beta) \in v^n_\Delta(\alpha) \tilde{\land} v^n_\Delta(\alpha)$. Analogously, we prove that $v^n_\Delta(\alpha \# \beta) \in v^n_\Delta(\alpha) \tilde{\#} v^n_\Delta(\alpha)$ for $\# \in \{\vee,\to\}$. Concerning negation, suppose that $v^n_\Delta(\alpha)=T_n$. Then, $\alpha \in \Delta$ and $\neg\alpha \notin \Delta$, and so $v^n_\Delta(\neg\alpha)=F_n \in \{F_n\}=\tilde{\neg} \, T_n=\tilde{\neg} \, v^n_\Delta(\alpha)$. If $v^n_\Delta(\alpha)=F_n$ the proof is analogous. Now, suppose that $v^n_\Delta(\alpha)=t^n_i$. Then, $\neg\alpha \in \Delta$ and so, by $(*)$,  $v^n_\Delta(\neg\alpha) \in D_n=\tilde{\neg} \, t^n_i=\tilde{\neg} \, v^n_\Delta(\alpha)$.  Finally, by $(*)$, if $v^n_\Delta(\obl\alpha) \in D_n$ then $\obl\alpha \in \Delta$  and so $\alpha \in \Delta\prim$ for all $\Delta\prim \in W_{can}^{(n)}$ such that $\Delta R_{can}^{(n)} \Delta\prim$, by item~10. of \cref{lemma:Cntruth}. This means that  $v^n_{\Delta\prim}(\alpha) \in D_n$ for all $\Delta\prim \in W_{can}^{(n)}$ such that $\Delta R_{can}^{(n)} \Delta\prim$, by $(*)$ once again, therefore 
$v_\Delta^n (\obl \alpha) \in D_n=\aobl \big(\{v_{\Delta\prim}^n(\alpha) \ : \ \Delta R_{can}^{(n)} \Delta\prim\}\big)$. Now, if $v^n_\Delta(\obl\alpha) = F_n$ then $\obl\alpha \notin \Delta$, by $(*)$, hence there exists some $\Delta\prim \in W_{can}^{(n)}$ such that $\Delta R_{can}^{(n)} \Delta\prim$ and $\alpha \notin \Delta\prim$, by item~10. of \cref{lemma:Cntruth}. This means that  $v^n_{\Delta\prim}(\alpha) = F_n$ for some $\Delta\prim \in W_{can}^{(n)}$ such that $\Delta R_{can}^{(n)} \Delta\prim$, by $(*)$ once again. From this, $v_\Delta^n (\obl \alpha) \in \{F_n\}=\aobl \big(\{v_{\Delta\prim}^n(\alpha) \ : \ \Delta R_{can}^{(n)} \Delta\prim\}\big)$.
\\[1mm]
(II) Let us see now that  each  $v^n_\Delta$ satisfies the restrictions imposed in  \cref{def:Cnrest} and \cref{def:Cnrest3}. Thus, assume first that $v^n_\Delta(\alpha) = t^n_0$. Then, $\alpha, \neg\alpha \in \Delta$ and $\alpha^1 \notin \Delta$. Hence, $(\alpha \land \neg \alpha) \in \Delta$, by item~1. of \cref{lemma:Cntruth}, and $\neg(\alpha \land \neg \alpha) = \alpha^1 \notin \Delta$. From this, $v^n_\Delta(\alpha \land \neg \alpha) = T_n$. Now, suppose that $v^n_\Delta(\alpha) = t^n_k$ for some $1 \leq k \leq n-1$. By \cref{def:Cn-bival-truth}, $\alpha, \neg\alpha \in \Delta$ and $\alpha^{k+1} \notin \Delta$. From this, $(\alpha \land \neg \alpha) \in \Delta$ and $\neg(\alpha \land \neg \alpha) = \alpha^1 \in \Delta$, by items~1. and~8. of \cref{lemma:Cntruth}. By  \cref{def:Cn-bival-truth},  $v^n_\Delta(\alpha \land \neg \alpha) \in I_n$. Suppose that $\neg(\alpha^1) \notin \Delta$. By item~6. of \cref{lemma:Cntruth} applied to $\alpha^1$,  it follows that $\alpha^2 \in \Delta$ and  $\neg(\alpha^2) \notin \Delta$. By applying iteratively item~6. of \cref{lemma:Cntruth} to $\alpha^2$, $\alpha^3$ and so on, we conclude that $\alpha^{k+1} \in \Delta$, a contradiction. This means that  $\neg(\alpha^1) \in \Delta$. Since $\alpha^1 \in \Delta$ and $(\alpha^1)^k=\alpha^{k+1} \notin \Delta$, we conclude by  \cref{def:Cn-bival-truth} that  $v^n_\Delta(\alpha^1) = t^n_{k-1}$. This shows that the restrictions of  \cref{def:Cnrest} are satisfied by the functions $v^n_\Delta$. Finally, suppose that $v^n_\Delta(\alpha) \in Boo_n$. By item~2 of \cref{coro:v_Delta-CnD}, $\alpha^{(n)} \in \Delta$. By  ({\bf P$\obl$}$_n$), $(\obl \alpha)^{(n)} \in \Delta$ and so, by item~2 of \cref{coro:v_Delta-CnD} once again, $v^n_\Delta(\obl \alpha) \in Boo_n$. This shows that the condition of \cref{def:Cnrest3} are also satisfied by the functions $v^n_\Delta$.

This concludes the proof.  
\end{proof}

   \begin{thm}[Completeness of $C^D_n$ w.r.t. swap Kripke models] \ \\
	For any set $\Gamma \cup \{\varphi\} \subseteq For(\Sigma^{C_1}_D)$, if $\Gamma \vDash_{C^D_n} \varphi$ then $\Gamma \vdash_{C^D_n} \varphi$.
    \end{thm}
    \begin{proof}
        Suppose that $\Gamma \nvdash_{C^D_n} \varphi$. Then, there is a $\varphi$-saturated set $\Gamma \subseteq \Delta$ such that $\varphi \notin \Delta$. From \cref{prop:can-model-CnD},  $\M_n$ is a swap Kripke model for $C^D_n$  and  $\Delta$ is a world in $\M_n$ such that $\M_n,\Delta \vDash\Gamma$ but $\M_n,\Delta \nvDash\varphi$. This implies that $\Gamma \nvDash_{C^D_n} \varphi$.
    \end{proof}

\subsection{A small addition}

We briefly mention that in order to validate $\axD$ in standard formulation, that is, using the primitive paraconsistent negation $\neg$ of $C_n$, we need one more restriction added to our valuations, namely:

\begin{defn}
    \label{def:Cnrest2}
     A swap Kripke model  for ${C}^D_n$ is said to be {\em strict} if the valuations satisfy, in addition, the following constraint:
         $$\mbox{If } v_w^n(\obl \alpha) \in D_n \ \mbox{ then, for every $w\prim \in W$ such that $w R w\prim$, } \ v^n_{w\prim}(\alpha)= T_n.$$
\end{defn}

Then it is easy to see that $\axD$ (formulated with $\neg$) is valid w.r.t.  strict  swap Kripke models  for ${C}^D_n$. 

\begin{prop}
   Axiom schema 
   $$({\bf SD}_n) \ \ \ \obl \alpha \to \neg\obl\neg\alpha$$ 
   is valid w.r.t. strict  swap Kripke models  for ${C}^D_n$. 
\end{prop}
\begin{proof}
    Let $\M$ be a  strict  swap Kripke model  for ${C}^D_n$, and suppose that, for some formula $\alpha$ and some world $w$ in $\M$, $v^n_w(\obl \alpha) \in D_n$ but $v^n_w(\neg\obl \neg\alpha) =F_n$. The latter implies that $v^n_w(\obl \neg\alpha) =T_n \in D_n$. By \cref{def:Cnrest2} it follows that, for every $w'$ in $\M$ such that $w R w'$, it is the case that $v^n_{w'}(\alpha)=T_n=v^n_{w'}(\neg\alpha)$. But this is a contradiction, since  $\tilde{\neg} \, T_n = \{F_n\}$. This shows that $\M \vDash \obl \alpha \to \neg\obl\neg\alpha$ for every strict  swap Kripke model  for ${C}^D_n$ and for every formula $\alpha$.
\end{proof}

It also interesting to notice that these strict models collapse the two notions of permission. Whereas in the first formulation of $C^n_D$ one would formally be able to characterize two distinct notions of permission, namely $\cneg^{(n)} \obl \cneg^{(n)}\alpha$ and $\neg \obl \neg\alpha$,  in this strict formulation there is a collapse, since whenever $\obl \alpha$ holds in $w$, in all the worlds accessible to $w$, $\alpha$ behaves classically and is true. It is an easy exercise to see the preservation of soundness and completeness of  ${C}^D_n$ plus $({\bf SD}_n)$ w.r.t. strict  swap Kripke models. 

\section{Applications of $C^D_1$ and $C^D_n$ to Moral Dilemmas}

Since most of the logics here presented have already been studied when applied to Chisholm's paradox \cite{Coniglio2009b} and to Moral Dilemmas \cite{VazMaruchi25}, we focus on discussing the applications of $C^D_1$ and $C^D_n$ to moral dilemmas. The special attention to moral dilemmas follows the argument presented in \cite{VazMaruchi25}, where the authors argue that paraconsistent deontic logics are better off applied to moral dilemmas, since they can further enlighten the subject and also deal with conflicting obligations more directly, while also giving an explanation to the phenomena. 

Shortly, moral dilemmas are usually stated in bimodal logics, having two operators, one for obligation, $\obl$, and another alethic operator, namely $\Diamond$, which denotes possibility in the actual world. Moral dilemmas can thus be formalized as $$\obl \alpha \land \obl \beta \land \neg \Diamond (\alpha \land \beta)$$ 
Clearly, when $\beta = \neg \alpha$, we have conflicting obligations. A standard example of moral dilemmas is Sophie's Dilemma, in which a prisoner of a Nazi camp has to decide to save either her daughter or her son, who are scheduled to be executed and if she decides not to pick between one of them, both are executed.  

There are a few remarks to be done here. First, is that the argument in \cite{VazMaruchi25} points towards the conflicting obligations being the root cause of the problem in such scenarios, thus relegating a secondary role to the alethic operator. Second, is that the negation appearing before the alethic operator could also be regarded as a different negation, which behaves classically, so that it would take away the possibility of solving the problem by paraconsistent logic alone. Our focus is to solve the problem deontically, with paraconsistency being a feature of the deontic systems we are studying. In other words, we want that the focus of our discussion are formulas that are deontic and have negations inside the scope of the modal operator.

Having laid down these considerations, we can overview how the systems $C^D_1$ and $C^D_n$ deal with moral dilemmas.

If we look at system $C^D_1$, there are a few options on how to solve moral dilemmas. The first and obvious route is to try and differentiate the negation happening in the scope of the modal operator. Thus, for example, if the formula is $\obl \alpha \land \obl \neg \alpha$, then this is perfectly acceptable in our model for $C^D_1$. However, this could be seen as \textit{not enough}. If we take $\obl \neg \alpha$ to represent a \textit{weak} prohibition, then nothing like the scenario of Sophie's Dilemma could be pictured. In fact, this could, in some sense, be pictured as a minor dilemma. For example, picture the following scenario: you have a class on Friday night and a friend calls you offering a ticket for a concert that they can not attend anymore due to personal reasons. By attending the concert, you miss class and potentially fail your course, but it happens that is a band you really like, and might be your last opportunity to see them live, and as a big fan of art, you have a principle to always support the artists you like whenever possible. This could count as a minor dilemma, since the consequences of this act would not have big consequences, such as somebody's death or the starting of a war. So this motivates us to analyze cases in which $\obl \alpha \land \obl \cneg \alpha$ occurs. When taking prohibition to be $\obl \cneg \alpha$, we are forced to analyze $\alpha$ classically in a deontic setting and thus explosion is recovered. Hence, in such scenarios, $C^D_1$ allows for ``weak'' dilemmas without explosion to occur, but not for strong ones, as in the case of Sophie's Dilemma.

Going up the hierarchy, that is, talking about all of $C^D_n$, it might be the case that these logics integrate even finer distinctions of ``strengths'' of permissible dilemmas. What then $C^D_n$ allows is that stronger dilemmas can be accounted for in the system. Since we move the classicality up the hierarchy, now the strong negation understood as $\neg \alpha \land \alpha^1$ can work as part of a definition of a ``strong obligation'', which can then account for such a distinction. In this sense, we would have the conflicting obligations in Sophie's Dilemma being assigned a designated value in the system without trivialization.

Another discussion necessary in order to bring these systems to their full potential is whether or not these distinctions between weak and strong obligations in fact play a role in the actual situations we are trying to formalize. This, however, is a discussion that the authors will delve into in further papers. 

In summary, while $C^D_1$ allows for some distinction between weak and strong prohibitions in a naive sense, it does not accommodate for a conflict between an obligation and a strong prohibition, thus limiting its usefulness to formalize and deal with moral dilemmas. On the other hand, it is just the first step of a whole hierarchy, which, in turn, allows for both cases to be formalized and dealt with, with trivialization.

\section{Concluding Remarks}

This paper presented a new way to give semantics to {\ldis} conflating possible worlds and nondeterministic endeavours via swap structures. Although our approach here is not fully deterministic, it maintains a good balance between Nmatrices, RNmatrices and Kripke semantics, showing that it is possible to mix them, by satisfactorily characterizing these logics in such a setting. In particular we described many logics along the {\ldis} hierarchy, beginning with {\DmbC}, the minimal \textit{LFI} equipped with the modal axioms for \sdl, and walking up the hierarchy basing our propositional semantics on Nmatrices.

By adding {\bf (cl)} to {\DmbC}, we strengthen our systems in such a way that it becomes impossible to characterize them in terms of finite Nmatrices alone. We then resort to a reading of RNmatrices adapted to swap structures, namely, a restriction in the admissible swap valuations. This move allows us to characterize {\DmbCcl} and also stronger logics, such as {\DCila}, and the whole of $C^D_n$ hierarchy.   

Regarding the latter, the developments here presented are entirely new. In the case of $C^D_1$, a sketch of its semantics, by means of bivaluations, was given in \cite{CostaCarnielli1986}. In this paper, we fully develop those proofs by means of the novel notion of  swap Kripke structures, giving proofs for both {\DCila} and $C^D_1$. For $C^D_n$ in general, it is the first time this family of systems is fully developed and semantically characterized. We also discuss briefly the different systems that can be defined given the multiple notions of negation these system are able to express. A thorough survey of such systems would require a paper on its own, and our objective here is to lay down the technical grounds upon which this discussion is allowed to be attained.

We believe this combination between nondeterministic semantics and possible worlds semantics can be fruitful in the conception of new semantics for logical systems, since it allows for the introduction of new concepts into the logic, for example, detaching modal notions from possible world semantics, thus allowing for a higher expressivity of these nondeterministic modal systems. The mix between them also allow for systems in which each world is nondeterministic and, as seen in the case for $C^D_n$, the modal operator that has its truth conditions based on these nondeterministic worlds inherits the nondeterministic behavior, as well as being sufficiently expressive as to accommodate for two notions of prohibition in its full capabilities. This allows  a more or less fine-grained distinctions of dilemmas, depending on how far into the hierarchy the dilemma is modeled. We also point out that further investigation should look deeper into the philosophical aspects of such systems, as well how they fare when modeling other paradoxes of deontic logics. 

\

\

\noindent{\bf Acknowledgements:} Vaz holds a PhD scholarship from the S\~ao Paulo Research Foundation (FAPESP, Brazil), grant 2022/16816-9, and was also financed by the German Academic Exchange Service (DAAD, Germany), under the Bi-nationally supervised/Cotutelle Doctorate degree program.
Coniglio acknowledges support by an individual research grant from the National Council for Scientific and Technological Development (CNPq, Brazil), grant 309830/2023-0. All the authors were supported by the S\~ao Paulo Research Foundation (FAPESP, Brazil), thematic project {\em Rationality, logic and probability -- RatioLog}, grant  2020/16353-3.

\

\bibliographystyle{plain}

\

\

%For articles with one author:
\noindent \textsc{Mahan Vaz}\\
{Instituto de Filosofia e Ciências Humanas (IFCH), and \\ Institut für Philosophie I, Logik und Erkenntnistheorie\\
Universidade Estadual de Campinas (UNICAMP), and\\ Ruhr-Universität Bochum\\
Universitätstraße 150, GB04/43, Bochum, Germany}\\
\texttt{Mahan.VazSilva@edu.ruhr-uni-bochum.edu}

\

\

%For articles with two authors
\noindent \textsc{Marcelo E. Coniglio}\\
{Instituto de Filosofia e Ci\^{e}ncias Humanas (IFCH), and\\
Centro de L\'ogica, Epistemologia e Hist\'oria da Ci\^encia (CLE)\\  Universidade Estadual de Campinas (UNICAMP)\\ Rua S\'ergio Buarque de Holanda, 251, Campinas, SP, Brazil}\\
\texttt{coniglio@unicamp.br}

\end{document}